%% file: uday_l4dc25.tex
\documentclass[final]{l4dc2025}

\usepackage[utf8]{inputenc} 
\usepackage[T1]{fontenc}    

\usepackage{times}
\usepackage{latexsym,amsfonts}
\usepackage{color}
\usepackage{colortbl}
  \graphicspath{{./Figures/}}
  \DeclareGraphicsExtensions{.jpeg,.png,.jpg,.eps,.pdf}
\usepackage{tikz}
\usepackage{epstopdf}
\usepackage{hyperref}
\usepackage{xcolor}

\usepackage{enumitem}
\setlist{  
  listparindent=\parindent,
  parsep=0pt,
}
\setlist[itemize]{leftmargin=*}

\usepackage[T1]{fontenc} 
\usepackage{times} 
\usepackage{comment}
\usepackage{bbm}
\usepackage{cancel}
\usepackage{multirow}
\input{notation.tex}

\usepackage{graphicx}
\usepackage{algorithm}
\usepackage[normalem]{ulem}

\usepackage{minitoc}

\definecolor{verm}{rgb}{0.6,0.2,0.2}
\definecolor{purp}{rgb}{0.3,0.1,0.6}
\definecolor{purple}{rgb}{0.4,0.0,0.6}
\definecolor{bggreen}{rgb}{0.1,0.3,0.1}
\definecolor{dgreen}{rgb}{0.1,0.6,0.1}
\definecolor{black}{rgb}{0.0,0.0,0.0}
\definecolor{crim}{rgb}{0.3,0.1,0.1}
\definecolor{dred}{rgb}{0.5,0.1,0.1}
\definecolor{NavyBlue}{HTML}{1f4e79}
\definecolor{DeepTeal}{HTML}{006d77}
\definecolor{DarkSlateGray}{HTML}{2f4f4f}
\definecolor{SlateBlue}{HTML}{6a5acd}
\definecolor{MutedBurgundy}{HTML}{800020}

{\bfseries}{\itshape}
{\bfseries}{\rmfamily}
\newtheorem{assumption}{Assumption}{\bfseries}{\itshape}
\newtheorem{corollaryA}{Corollary}{\bfseries}{\itshape}
\newtheorem{theoremA}{Theorem}{\bfseries}{\itshape}

\title[Nonconvex System IDentification]{Nonconvex Linear System Identification with Minimal State Representation}

\author{
 \Name{Uday Kiran Reddy Tadipatri} \Email{ukreddy@seas.upenn.edu}\\
 \addr University of Pennsylvania, USA
 \AND
 \Name{Benjamin D. Haeffele} \Email{haeffele@seas.upenn.edu}\\
 \addr University of Pennsylvania, USA
 \AND
 \Name{Joshua Agterberg} \Email{jagt@illinois.edu}\\
 \addr University of Illinois Urbana-Champaign, USA
 \AND
 \Name{Ingvar Ziemann} \Email{ingvarz@seas.upenn.edu}\\
 \addr University of Pennsylvania, USA
 \AND
 \Name{Ren{\'e} Vidal} \Email{vidalr@seas.upenn.edu}\\
 \addr University of Pennsylvania, USA
}

\begin{document}
\maketitle

\begin{abstract}
Low-order linear System IDentification (SysID) addresses the challenge of estimating the parameters of a linear 
dynamical system from finite samples of observations and control inputs with minimal state representation.
Traditional approaches 
often utilize Hankel-rank minimization, which relies on convex relaxations that can require numerous, costly singular value decompositions (SVDs) to optimize.
In this work, we propose two nonconvex reformulations to tackle low-order SysID
(i) Burer-Monterio (BM) factorization of the Hankel matrix for efficient nuclear norm minimization, and
(ii) optimizing directly over system parameters for real, diagonalizable systems  with an atomic norm style decomposition. 
These reformulations circumvent the need for repeated heavy SVD computations, significantly improving computational
efficiency. Moreover, we prove that optimizing directly over the system parameters yields lower statistical
error rates, and lower sample complexities that do not scale linearly with trajectory length
like in Hankel-nuclear norm minimization.
Additionally, while our proposed formulations are nonconvex, we provide theoretical guarantees of achieving global optimality in polynomial time. Finally, we demonstrate algorithms that solve these nonconvex programs and validate our theoretical claims on synthetic data.
\end{abstract}

\begin{keywords}
System Identification, Hankel rank minimization, Sample Complexity.
\end{keywords}

\section{Introduction}
We consider the linear time-invariant system \eqref{eq:ls} with time $t \in \mathbb{N}$, $\x_t \in \R^{n_x}$ as hidden states, $\u_t \in \R^{n_u}$ as control inputs,
$\y_t \in \R^{n_y}$ as outputs, $\zeta_t \in \R^{n_y}$ as output noise,
and system parameters $(A \in \R^{n_x \times n_x}$, $B \in \R^{n_x \times n_u}$, 
$C \in \R^{n_y \times n_x}$, $D \in \R^{n_y \times n_u})$ described as
\be\label{eq:ls}\tag{LS}
\begin{split}
&\x_{t+1} = A\x_{t} + B\u_t; \x_1 = 0,\\
&\y_t = C\x_t + D\u_t + \zeta_t.
\end{split}
\ee

Linear \textit{System IDentification} (SysID) aims to estimate these parameters
$(A$, $B,$ $C,$ $D)$ using finite rollouts with finite length trajectories,
$\{(\u_t^{i}, \y_t^{i})_{t=1}^{2L+2}\}_{i=1}^{N}$. However, the true system order or state dimension $n_x$ is 
typically unknown.
It is of practical interest to find systems that are \textit{minimal order}
as such models provide faster response times, simplified control designs, and improved robustness, 
while mitigating overfitting during the learning phase.
This task can be formulated as a dimension minimization problem under an $L_2$ norm constraint (see Program \eqref{eq:P0}). 
However, solving this problem is known to be NP-HARD \citep{vandenberghe-boyd-siam96}.

In classical linear system theory, the order of the system is determined by the rank of the 
Hankel matrix constructed from the Markov parameters $(CB, CAB, \dots, CA^{2L+1}B)$
\cite[Corollary 6.5.7]{sontag2013mathematical}. Convex relaxations, such as minimizing the Hankel
nuclear norm instead of the rank, have been proposed to make this problem computationally tractable
\citep{fazel-et-acc03}. However,
this presents two key challenges: 
(i) \textbf{Computational complexity}: The number of parameters scales 
linearly with the trajectory length $L$, making nuclear norm minimization increasingly difficult.
(ii) \textbf{Statistical inefficiency}: The high dimensionality leads to worse statistical error rates 
and larger sample complexity requirements \citep{oymak2019non}.

In the literature on matrix sensing, nuclear norm minimization is tackled via nonconvex reparameterization 
using its variational form with Frobenius norms \citep{burer-monteiro-mp03}.
This reformulation enables the use of efficient first-order optimization methods,
such as gradient descent and Nesterov acceleration \citep{nesterov-nag83}, which converge
in $\mathcal{O}(1/T)$, and $\mathcal{O}(1/T^2)$, respectively, where $T$
is the number of iterations. 
In contrast, direct nuclear norm minimization often requires computing proximal operators \citep{parikh-boyd-fto14} that 
depend on SVD, making it less computationally efficient.
In this work, we propose two reformulations, the first of which is Hankel nuclear norm minimization for SysID 
(see Program \eqref{eq:P2}) using the BM factorization \citep{burer-monteiro-mp03}.
Despite being nonconvex, this formulation eliminates the need for frequent SVD computations, significantly enhancing computational efficiency, and can still be provably globally optimized in polynomial time.

While this first reformulation offers clear advantages in terms of optimization efficiency, the statistical error rates 
and sample complexities remain comparable to those of Hankel nuclear norm minimization. In particular, the statistical error rates scale with trajectory length, despite the number of system parameters remaining constant. 
Our second reformulation (see Program \eqref{eq:P1}) proposes a method of SysID of a real, diagonalizable system
where we perform optimization directly over the system parameters space $(A, B, C, D)$ with a structured 
regularization similar to an atomic norm decomposition.  Although this reformulation is also nonconvex,
we show that it can be efficiently optimized to global optimality. Furthermore, this approach achieves improved 
statistical error rates and reduced sample complexity compared to existing Hankel nuclear 
norm minimization heuristics.

\textbf{Contributions.}
Our key contributions are as follows:
\vspace{-15pt}
\begin{enumerate}
    \item \textbf{Nonconvex Relaxations for low-order SysID.} We propose two nonconvex reformulations (i) for 
general linear systems via BM re-parametrization \eqref{eq:P2} of the
Hankel matrix, (ii) by directly optimizing over
system parameters \eqref{eq:P1}, restricted to real, diagonalizable systems.
These nonconvex problems are first of their kind
and serve as relaxations to the NP-HARD problem \eqref{eq:P0}.
    \item \textbf{Global Optimality Guarantees.} We provide guarantees of global optimality for the two nonconvex reformulations when they are solved using first-order optimization procedures.
    Furthermore, we affirm that these nonconvex problems can be solved in polynomial time.
    The second reformulation overcomes the quadratic dependency on trajectory length in the
    computational efficiency of existing methods.
    \item \textbf{Statistical Error Rates and Sample Complexities.} We provide statistical error rates and sample complexities for both the reformulations
    while the first reformulation scales poorly compared to the second
    due to its linear dependency on trajectory length. 
    \item \textbf{Numerical evidence.} We propose algorithms to tackle these
    nonconvex programs, provide numerical simulations that corroborate our 
    theoretical insights. Under a fixed compute budget, our algorithms outperform existing methods.
\end{enumerate}
These results highlight that for real, diagonalizable systems
the reformulation \eqref{eq:P1} achieves both superior statistical efficiency and
reduced sample complexity compared to \eqref{eq:P2},
\cite{sun-et-al-ojcs22}, and \cite{lee2022improved}.

\textbf{Organization.} First, in \textsection\ref{sec:prob_form} we introduce the problem and mathematical
formulations. Next, we move onto \textsection\ref{sec:opt_results} where we present 
optimality certificates for the nonconvex formulations. Then, in \textsection\ref{sec:stat_results}
we present statistical error rates and sample complexities for each of the formulations.
Later, we move to \textsection\ref{sec:numerical_simulations} that presents numerical experimentation
and comparison with existing methods. Finally, we conclude and discuss the future work 
in \textsection\ref{sec:conclusion}. 
Further related works and proofs for all the mathematical statements can be found
in Appendix of our arXiv version.\\
\textbf{Notation.}   
For a matrix $A$, $[A]_{i_1:i_2, j_1:j_2}$ denotes its submatrix from rows $i_1$--$i_2$ and columns 
$j_1$--$j_2$; $[[A_{i,j}]_{i \in \mathcal{I}}]_{j \in \mathcal{J}}$ stacks blocks $A_{i,j}$ row- and 
column-wise over $\mathcal{I}$ and $\mathcal{J}$. The superscript $\dagger$ denotes the adjoint for
operators and the Moore--Penrose pseudoinverse for matrices. 
The inequality $f(x) \gtrsim g(x)$ or $f(x) \geq \mathcal{O}(g(x))$ means there exists a positive
constant $c$ such that $f(x) \geq c g(x)$ for all $x$. For some $\alpha > 0$, $(\ln x)^{\alpha}$
is denoted as ${\sf polylog}(x)$, and $\tilde{\mathcal{O}}(x)$ denotes
$\mathcal{O}(x \cdot {\sf polylog}(x))$.

\section{Problem Formulation}\label{sec:prob_form}
Given a noisy linear time-invariant system \eqref{eq:ls},
we are interested in special approximation of outputs $\{\y_t\}_{t=1}^{2L+2}$ for given inputs $\{\u_t\}_{t=1}^{2L+2}$
that has low-order, or smaller state dimension, $n_x$. Formally, define
impulse response \footnote{Due to the trivial estimation of $D$ we often omit its dependency in the impulse response and denote it as $G(A, B, C)$.} 
$G(A, B, C, D)$ $\in$ $\R^{2(L+1)n_y \times 2(L+1)n_u}$ where $i, j \in [2(L+1)]$, and
\vspace{-10pt}
\be
[G(A, B, C, D)]_{i, j} = \begin{cases}
0 & i < j,\\
D & i = j,\\
CA^{i-j-1}B & i > j.
\end{cases} \in \R^{n_y \times n_u}.
\ee
Our goal is to solve the optimization
problem \eqref{eq:P0},
\vspace{-3pt}
\be\label{eq:P0}
\begin{split}
\min_{n_x, A, B, C, D} n_x , 
\quad \text{such that }\frac{1}{4N(L+1)}\sum_{i=1}^{N}
\nmm{\vecc(Y_i)- G(A, B, C, D)\vecc(U_i)}_F^2 \leq \epsilon,
\end{split} \tag{P0}
\ee
where $Y_i := \begin{bmatrix}\y_1^{i} & \y_2^{i} & \dots & \y_{2(L+1)}^{i} \end{bmatrix} \in \R^{n_y \times 2(L+1)}$,
$U_i := \begin{bmatrix}\u_1^{i} & \u_2^{i} & \dots & \u_{2(L+1)}^{i} \end{bmatrix} \R^{n_y \times 2(L+1)}$,
$\vecc: \R^{m \times n} \to \R^{m\cdot n \times 1}$ is column-wise stacking operator.
It is well known that solving \eqref{eq:P0} in general
is NP-HARD \citep{vandenberghe-boyd-siam96}. Numerous convex relaxations of \eqref{eq:P0} have been proposed to tackle this issue under certain conditions (beyond the scope of this work), such as minimizing the nuclear norm of the Hankel matrix
\citep{fazel-et-acc03, recht-et-al-cdc08, recht-et-al-siam10}. To place this in context, define the Hankel operator $\mathcal{H} : \R^{n_y \times (2L+1)n_u}
\to \R^{(L+1)n_y \times (L+1)n_u}$ for a block sequence $\{K_{t}\}_{t \in [2L+1]} \subset \R^{n_y \times n_u}$ as
\vspace{-5pt}
\be
\forall i, j \in [L+1], \left[\mathcal{H}\left(\begin{bmatrix}K_1 & K_2 & \dots & K_{2L+1}\end{bmatrix}\right)\right]_{1 + (i-1)\cdot n_y:i\cdot n_y, 1+(j-1)\cdot n_u:j\cdot n_u}
= K_{i+j-1}.
\ee
\vspace{1pt}
In \cite{sun-et-al-ojcs22} the authors proposed to solve the reformulated problem
\vspace{-5pt}
\be\label{eq:P1_1}
\tag{P1'}
\min_{H, D}\frac{1}{2N}\sum_{i=1}^{N}
\nmm{\y_{2L+1}^{i} - \begin{bmatrix}H & D\end{bmatrix}\vecc(U_i)}_F^2 + \l \nmm{\mathcal{H}(H)}_*.
\ee
After obtaining the Markov parameters $H$ they perform the
Ho-Kalman procedure \citep{ho-kalman-66} to obtain the 
system parameters $(A, B, C, D)$. The matrix $H$ has 
$2(L+1)n_yn_u$ number of unique parameters to estimate.
The time complexity of SVD on the matrix $\mathcal{H}(H)$ is 
$\mathcal{O}(L^2n_yn_x)$ \citep{xu-qiao-laa08}, 
so the quadratic dependence on $L$ it causes a computational bottleneck when $L$ is large.
Therefore, we propose adopting a BM type re-parameterization \citep{burer-monteiro-mp03} in two ways: (i) by learning two low-rank matrices whose product forms a Hankel matrix, and (ii) by directly estimating the system parameters $(A, B, C, D)$. First, 
recall the adjoint of the Hankel operator $\mathcal{H}^{\dagger} : \R^{(L+1)n_y \times (L+1)n_u} \to \R^{n_y \times (2L+1)n_u}$ as
\be
\mathcal{H}^{\dagger}\left(\begin{bmatrix}[K_{i+j-1}]_{i \in [L+1]}
\end{bmatrix}_{j \in [L+1]}\right) = \begin{bmatrix}K_1 & K_2 & \dots & K_{2L+1}\end{bmatrix}
\in \R^{n_y \times (2L+1)n_u},
\ee
and note that $\mathcal{H} \circ \mathcal{H}^{\dagger} = {\sf id}$, and 
$\mathcal{H}^{\dagger} \circ \mathcal{H} = {\sf id}$, where ${\sf id}$ is the identity operator.

\textbf{First formulation}. We define the optimization program \eqref{eq:P2} as
\vspace{-8pt}
\be\label{eq:P2}
\min_{n_x, V, Z, D} \frac{1}{2N}\sum_{i=1}^{N}\nmF{\y_{2L+1}^{i} - \begin{bmatrix}\mathcal{H}^{\dagger}(VZ^T) & D\end{bmatrix} \vecc(U_i)}^2 + \frac{\l}{2}\left[\nmF{V}^2 + \nmF{Z}^2\right],
\tag{P1}
\ee
where $V \in \R^{(L+1)n_y \times n_x}$, and $Z \in \R^{(L+1)n_u \times n_x}$.
Let the impulse response for this parameterization be 
$G'(V, Z) \in \R^{2(L+1)n_y \times 2(L+1)n_u}$.

Program \eqref{eq:P2} redefines the nuclear norm minimization via the sum of 
squared Frobenius norms $\frac{1}{2}\left[\|V\|_F^2 + \| Z \|_F^2\right]$. In contrast to the nuclear norm  regularization, the Frobenius norm minimization has time
complexity $\mathcal{O}(L(n_y + n_u))$, which is only linear in $L$. 
However, in both Programs \eqref{eq:P1_1} and \eqref{eq:P2}, the number of parameters
scales with the trajectory length. For \eqref{eq:P1_1}, it has been shown that the 
sample complexity is $NL \gtrsim \tilde{\mathcal{O}}(L^2 n_u (n_y + n_u))$ 
\citep{sun-et-al-ojcs22}. In contrast, we will see in Theorem 
\ref{thm:thm_p2_stat} \eqref{eq:P2} achieves a smaller sample complexity of 
$NL \gtrsim \tilde{\mathcal{O}}(Ln_u (n_y + n_u))$.
Nevertheless, there remains room for improvement in the sample complexity, 
given that we have access to $\mathcal{O}(NL)$ data points.
We consider another reformulation that ameliorates this suboptimal trajectory dependence. Unfortunately, the theoretical guarantees necessitate that the systems 
be real, diagonalizable, which imposes limitations on the class of realizable systems 
\citep{fernandez-et-al-esiam15}. However, previous studies have shown that
estimation of non-diagonalizable systems is highly challenging 
\citep{tu-et-al-jmlr24} which is an avenue for future work.

\textbf{Second formulation}. For real, diagonalizable systems
we optimize
over the system parameters and dimension directly via a regularization that can induce
low-order structure. We achieve this by considering the program
\be\label{eq:P1}
\tag{P2}
\min_{n_x, \a, B, C, D}\frac{1}{4N(L+1)}\sum_{i=1}^{N}
\nmm{\vecc(Y_i)- G({\sf diag}(\a), B, C, D)\vecc(U_i)}_F^2 + \l \Theta(\a, B, C, D),
\ee
where $\a \in \R^{n_x}$, and $\Theta(\a,B,C,D)$ is a regularizer that we will define subsequently 
in \textsection\ref{sec:opt_results}. 
Since we are dealing with real, diagonalizable systems the program \eqref{eq:P1} optimizes only
on the spectrum of state-transition matrix $A$, and eigen matrix is assumed to be absorbed in $B$, and $C$.
Although Program \eqref{eq:P1} is nonconvex, specific choices of 
$\Theta(\a, B, C, D)$ allow us to provide a certificate of optimality when employing 
first-order optimization methods.
In this work, we provide certificates for global optimality and analyze statistical 
recovery errors
for the programs \eqref{eq:P2}, and \eqref{eq:P1} for a trivial feed-through matrix, $D = 0$. The bulk of the work (in \textsection \ref{sec:opt_results}) relies in re-writing the programs
\eqref{eq:P2}, and \eqref{eq:P1} as sums of slightly generalized positively homogenous functions
for which global optimality guarantees are well studied in \cite{haeffele-vidal-tpami20}. In \textsection \ref{sec:stat_results}
we instantiate the approach of \cite{ziemann-tu-nips22} to provide tight 
(up-to log factors) error rates and sample complexities.

\section{Optimality certificates of nonconvex programs}\label{sec:opt_results}
In this section, we present certificates for the global optimality 
to each of the formulations \eqref{eq:P2}, and \eqref{eq:P1},
whose proofs can be found in \textsection\ref{sec:apdx_opt}. First, we state
Proposition \ref{prop:prop_p2} that establishes optimality 
of any stationary points of \eqref{eq:P2}.

\begin{proposition}\label{prop:prop_p2}
Let $(U_i, Y_i)$ be $N$ roll-outs of the system \eqref{eq:ls}. Consider the
estimator via factors $\hat{V}$ $\in$ $\R^{(L+1)n_y \times n_x}$, $\hat{Z}$ $\in$ $\R^{(L+1)n_u \times n_x}$. Define $U_i' := [U_i]_{1:2L+1}$, and 
\be\label{eq:polar_p2}
M := \frac{1}{N\l}\sum_{i=1}^{N}\left(\y_{2L+2}^{(i)} - \mathcal{H}^{\dagger}(\hat{V}\hat{Z}^T)\vecc(U_i')\right){\vecc(U_i')}^T;
\text{\textsf{Polar}}^{\eqref{eq:P2}}:=\nmm{\mathcal{H}(M)}_2.
\ee
Suppose $\hat{V}, \hat{Z}$ are any stationary points of Program \eqref{eq:P2}. 
If $\text{\textsf{Polar}}^{\eqref{eq:P2}} = 1$, then $\hat{V}$ and $\hat{Z}$ are globally optimal. 
Otherwise, the objective can be reduced 
by augmenting the top-singular vectors $(\v^*, \z^*)$
of $\mathcal{H}(M)$ to 
$\begin{bmatrix}\hat{V} & \tau^*\v^*\end{bmatrix}$ $\in$ $\mathbb{R}^{(L+1)n_y \times (n_x+1)}$ and 
$\begin{bmatrix}\hat{Z} & \tau^*\z^*\end{bmatrix}$ $\in$ 
$\mathbb{R}^{(L+1)n_u \times (n_x+1)}$, for some scaling factor $\tau^*$.
\end{proposition}

\begin{corollaryA}\label{crl:crl_p2}
Under the conditions of Proposition \ref{prop:prop_p2}, if $(\hat{n}_x, \hat{V}, \hat{Z})$ are the globally
optimal points of program \eqref{eq:P2}, then the system has the order 
$\hat{n}_x$ and system parameters take the form
\be
\hat{A} = \left(\left[\hat{V}\right]_{1:n_y, :}\right)^{\dagger}\left[\hat{V}\right]_{1+n_y:2n_y, :},
\hat{B} = \left[\hat{Z}^T\right]_{:, 1:n_u},
\hat{C} = \left[\hat{V}\right]_{1:n_y, :}.
\ee
\end{corollaryA}

\textbf{Remarks:}
From Proposition \ref{prop:prop_p2} we can utilize any
first-order algorithms such as gradient descent, Polyak's momentum method 
\citep{polyak-ussr64}, or Nesterov accelerated method \citep{nesterov-nag83}
to approximately reach stationary points 
in polynomial time, we assume exact convergence for technical convenience. 
Then from the Equation \ref{eq:polar_p2}
we need to verify the condition $\text{\textsf{Polar}}^{\eqref{eq:P2}}$ = $1$. For which
we need to compute the Hankel norm that is computationally expensive requiring 
$\mathcal{O}(Ln_yn_u)$
iterations \citep{cariow-gliszczy-er12}.
Moreover, $\text{\textsf{Polar}}^{\eqref{eq:P2}}$ can
never be strictly less than 1, 
see \cite{haeffele_global_2017} for further discussion. 
However, for  program \eqref{eq:P1} we will see that such optimality check
is faster. To estimate
the system parameters we can perform a pseudo inverse on the learned factors, $V, Z$ as presented in Corollary 
\ref{crl:crl_p2}. By the variational nuclear norm re-paramterization, we have eliminated
the need for separate Ho-Kalman procedure that was needed in many Sub-space recovery algorithms like N4SID \citep{van-automatica94}. 

Before we state our next results we define
$\gamma(a) := \sum_{t=0}^{L}a^{2t}$, $P(a)$ $\in$ $\R^{2(L+1) \times 2(L+1)}$
where
\be\label{eq:reg}
[P(a)]_{i, j} = \begin{cases}
a^{i-j-1}\text{ if }i>j \\
0\text{ otherwise }
\end{cases} \in \R
\text{, and }
\Theta(\a, B, C, D) := \sum_{j=1}^{n_x}\gamma(a_j)\nmm{\b_j}_2\nmm{\c_j}_2.
\ee
Regularization $\Theta(\a, B, C, D)$ resembles the atomic norm type norm
considered in matrix factorization problems \citep{bach-arxiv13, haeffele_global_2017}
whose optimality guarantees are well-studied. We next present optimality guarantees for 
applying this atomic-norm-type regularization directly to system parameters in the context 
of low-order SysID.

\begin{theoremA}\label{thm:thm_p1_opt}
Let $(U_i, Y_i)$ be $N$ roll-outs from the system \eqref{eq:ls} with $\x_0 = 0$. 
Consider the
estimator of system parameters via $\big({\textsf{diag}}\left(\a\right)$, $\begin{bmatrix}\b_1, \b_2, \dots, \b_{n_x}\end{bmatrix}^T, \begin{bmatrix}\c_1, \c_2, \dots, \c_{n_x}\end{bmatrix}\big)$, of order $n_x$ where $\a \in \R^{n_x}$, $\b_{\cdot} \in \R^{n_u}$, $\c_{\cdot} \in \R^{n_y}$. 
Define $U_i' := [U_i]_{1:2L+1}$, and 
\be\label{eq:polar_p1}
M(a)\hspace{-2pt}:= \frac{1}{2N(L+1)\l}\sum_{i=1}^{N}\left[\hspace{-1pt}Y_i - \sum_{j=1}^{n_x}\c_j\b_j^TU_i'P^T(a_j)\hspace{-2pt}\right]\hspace{-1pt}\frac{P(a)}{\gamma(a)}{U_i'}^{T}\text{, }
\text{\textsf{Polar}}^{\eqref{eq:P1}}\hspace{-2pt}:= \sup_{a \in {\R}} \nmm{M(a)}_2.
\ee
Let $\{a_j, \c_j, \b_j\}_{j=1}^{n_x}$ be any stationary points of program \eqref{eq:P1}.
If $\text{\sffamily Polar}^{\eqref{eq:P1}} = 1$, then $\{a_j, \c_j, \b_j\}$ are global optimal points. Otherwise, we can
reduce the objective by the parameters 
$\{a_j, \c_j, \b_j\}_{j=1}^{n_x}$ $\cup$ $\{\tau_1a^*$ $,\tau\c^*,$ $\tau\b^*\}$ for
some scalings, $\tau_1, \tau$, where
$a^*$ is the supremizer of $\kappa$ and $\c^*, \b^*$ are the top-singular vectors the matrix $M(\a^*)$.
\end{theoremA}

\textbf{Remarks:}
Similar to the discussion in Proposition \ref{prop:prop_p2}, we can use any first-order
optimization methods to approximately reach stationary points in polynomial time. Once we reach stationary
points, this leaves us to check the condition $\textsf{Polar}^{\eqref{eq:P1}} = 1$, this requires
us to maximize the rational matrix polynomial, this can be done through line search
methods in one-dimension such as golden-section or Fibonacci search in polynomial time \citep{ben-nemirovski-lecture01}.

\textbf{Proof Strategy:}
Our proofs for Proposition \ref{prop:prop_p2}, and Theorem \ref{thm:thm_p1_opt} rely using the general framework
for global optimality by \cite{haeffele_global_2017}, who studied
global optimality guarantees for the objective of the form,
\vspace{-10pt}
\bd
\min_{r, \{W_j\}_{j=1}^{r}}\ell\left(Y, \sum_{j=1}^{r}\phi(W_j; X)\right) + \l \sum_{j=1}^{r}\theta(W_j).
\ed
where $\phi(\cdot)$, and $\theta(\cdot)$ are positively homogeneous with the same degree.

For program \eqref{eq:P2}, \eqref{eq:P1} we can re-write the
predictions and regularization as
\vspace{-10pt}
\be\label{eq:factor_maps}
{\Phi_{n_x}}^{\eqref{eq:P2}}(V, Z; U) := \sum_{j=1}^{n_x}\mathcal{H}^{\dagger}(\v_j\z_j^T)\vecc(U');
{\Theta_{n_x}}^{\eqref{eq:P2}}(V, Z) := \frac{1}{2}\sum_{j=1}^{n_x}\nmm{\v_j}_2^2 + \nmm{\z_j}_2^2, 
\ee
\be\label{eq:factor_maps_1}
{\Phi_{n_x}}^{\eqref{eq:P1}}(\a, B, C; U)\hspace{-2pt}:=\hspace{-2pt}\sum_{j=1}^{n_x}\left[P(a_j) \otimes \c_j\b_j^T\right]\vecc(U');
{\Theta_{n_x}}^{\eqref{eq:P1}}(\a, B, C)\hspace{-2pt}:=\hspace{-2pt}\sum_{j=1}^{n_x}\gamma(a_j)\nmm{\b_j}_2\nmm{\c_j}_2.
\ee
The pair $\left({\Phi_{n_x}}^{\eqref{eq:P2}}(\cdot), {\Theta_{n_x}}^{\eqref{eq:P2}}(\cdot)\right)$ satisfy homogeneity property which
enables use the results of \cite{haeffele_global_2017}. However, the pair $\left({\Phi_{n_x}}^{\eqref{eq:P1}}(\a, B, C; U),
{\Theta_{n_x}}^{\eqref{eq:P1}}(\a, B, C)\right)$ are positively homogeneous with respective to only $B$ and $C$ but
not $\a$. Nevertheless, this technical challenge was resolved in \cite{tadipatri-et-al-iclr25} by relaxing to a 
weaker positive homogeneity property.

\section{Statistical error rates and Sample complexities for nonconvex programs}\label{sec:stat_results}
In this section, we present statistical error rates and sample complexities
for each of the formulations \eqref{eq:P2}, and \eqref{eq:P1}. The high-level
proof strategy is same for Theorem \ref{thm:thm_p2_stat}, and \ref{thm:thm_p1_stat}
which is discussed near the end of this section for full proofs see \textsection\ref{sec:apdx_stat}.
To begin, we state few assumptions
that are required for statistical recovery of the system parameters. We impose
tail conditions on the inputs and noise.
\begin{assumption}[Data Model]\label{ass:data_model}
The control inputs, $\u_t \in \R^{n_u}$ are drawn independently
from sub-Gaussian distribution
with a proxy variance $\sigma_U^2/n_u$ and zero mean,
i.e., for any unit vector $\v \in \R^{n_u}$ and $\forall \l \geq 0$
$\mathbb{E}\left[\exp(\l\IP{\u_t}{\v})\right] \leq 
\exp\left({\l^2\sigma_U^2}/{2n_u}\right)$.
Let the covariance of $\u_t$ be $\Sigma_U \succeq 0$, and
define block matrix
$\tilde{\Sigma}_{U} = I_{2(L+1)} \otimes \Sigma_U$.
The outputs are governed by the system \eqref{eq:ls} with order $n_x^*$,
and $\zeta_{t} | \u_{1:t}$ is conditionally independent sub-Gaussian distribution
with a proxy variance $\sigma^2/n_y$. Denote the
impulse response for this linear system as $G^*$.
\end{assumption}
Our next assumption ensures compactness of the parameter class.
\begin{assumption}[Bounded parameters for \eqref{eq:P2}]\label{ass:local_lip_hankel}
The learned parameters lie in the parametric class defined by 
\be
\F_{\theta}^{\eqref{eq:P2}} := \left\{(V, Z): n_x \in \mathbb{N},
\nmm{\v_j}_2 \leq B_v, \nmm{\z_j}_2 \leq B_z\right\}.
\ee
\end{assumption}
Before we state Theorem \ref{thm:thm_p2_stat} we define optimal regularizer
\be\label{eq:opt_reg_P2}
\Omega^{\eqref{eq:P2}}(\hat{G}) := \inf_{n_x, V, Z \in \F_{\theta}^{\eqref{eq:P2}}}\frac{1}{2}\left[\nmF{V}^2 + \nmF{Z}^2\right],\text{ s.t. }G'(V, Z) = \hat{G},
\ee
and let $E_t \in \R^{n_y \times 2(L+1)n_y}$
such that $[E_t]_{1+(t-1)n_y:tn_y} = I_{n_y}$,
and rest of them to be zero. Now define condition number as 
\be
{\sf cond}_{\F_{\theta}^{P}} :=
\sup_{\theta \in \F_{\theta}^{P}}\sup_{t \in [2(L+1)]}
\frac{\sqrt{\mathbb{E}\left[\nmm{E_t\Phi_{n_x}^{P}(\theta; U)}_F^4\right]}}{\mathbb{E}\left[\nmm{E_t\Phi_{n_x}^{P}(\theta; U)}_F^2\right]}; P \in \{\eqref{eq:P2}, \eqref{eq:P1}\}.
\ee
Note that $\Phi_{n_x}^{P}(\theta; U)$ 
(see Equations \eqref{eq:factor_maps}, and \eqref{eq:factor_maps_1}) is linear in $U$. 
Therefore, if $U$ follows a Gaussian distribution, the 
condition number ${\sf cond}_{\F_{\theta}^{P}}$ evaluates to 
3. For sub-Gaussian distributions, this can be bounded 
using Proposition 6.1 from \cite{ziemann_tutorial_2024}.

\begin{theoremA}\label{thm:thm_p2_stat}
Let $(U_i, Y_i)$ be $N$ i.i.d roll-outs following the Assumptions \ref{ass:data_model}.
Fix a $\delta \in (0, 1]$. 
Suppose the regularization parameter is such that
$\l \leq \tilde{\mathcal{O}}\left(\frac{n_x(n_y + n_u)}{N} + \frac{\ln(1/\delta)}{NL}\right)$.
For any global optimal points $(n_x, \hat{V}, \hat{Z})$ of program \eqref{eq:P2} satisfying Assumption \ref{ass:local_lip_hankel}.
If $N/\ln(NL) \gtrsim {\sf cond}_{\F_{\theta}^{\eqref{eq:P2}}}^2 \times [Ln_x(n_y + n_u) + \ln(1/\delta)]$, then
w.p at-least $1-\delta$ we have that
\be\label{eq:stat_p2}
\begin{split}
\nmm{(G'(\hat{V}, \hat{Z})-G^*){\tilde{\Sigma}_{U}}^{1/2}}_{F}^2
\leq \tilde{\mathcal{O}}\left(\left(\sigma^2  + \Omega^{\eqref{eq:P2}}(G^*)\right)\left[\frac{n_x(n_y + n_u)}{N} + \frac{\ln(1/\delta)}{NL}\right]\right).
\end{split}
\ee
\end{theoremA}

\begin{corollaryA}\label{crl:crl_p2_stat}
Under the conditions of Theorem \ref{thm:thm_p2_stat},
the sample complexity for the recovery of $G^*$
through program \eqref{eq:P1} is
$NL \gtrsim (Ln_x(n_y + n_u) + \ln(1/\delta)) \cdot {\sf polylog}(NL)$.
\end{corollaryA}

\textbf{Remarks.} For a fixed failure rate in Equation \eqref{eq:stat_p2} we recover near tight
statistical error rates, na\"{\i}vely we have {\sffamily recovery error $\lesssim \sqrt{\text{\# parameters}/\text{\# samples}}$}.
This is optimal up-to a logarithmic factor in comparison to least-squares error. From Corollary \ref{crl:crl_p2_stat}
we infer that the sample complexity grows nearly linearly with the trajectory length, i.e.,
$NL \geq \tilde{\mathcal{O}}(Ln_x(n_y + n_u))$. This
recovers the
program \eqref{eq:P1_1} studied in \citep{sun-et-al-ojcs22}. 
In practice, dependency on the
trajectory length is undesirable because each trajectory in itself provides
data points albeit dependent ones.

However, this limitation does not apply to formulation \eqref{eq:P1}.
We next present Theorem \ref{thm:thm_p1_stat} which provide error rates and sample complexities for 
program \eqref{eq:P1}, effectively addressing the aforementioned issue. To set the stage, we 
begin by introducing a compactness assumption on the system parameters, similar to 
Assumption \ref{ass:local_lip_hankel}.

\begin{assumption}[Bounded parameters for \eqref{eq:P1}]\label{ass:local_lip}
The learned parameters lie in the parametric class defined by
\be
\F_{\theta}^{\eqref{eq:P1}} := \left\{(\a, B, C): n_x \in \mathbb{N}, |a_j| \leq B_a,
\nmm{\b_j}_2 \leq B_b, \nmm{\c_j}_2 \leq B_c\right\}.
\ee
Furthermore for all
$(\a', B', C'), (\a, B, C) \in \F_{\theta}^{\eqref{eq:P1}}$ and some constant $k_a$ dependent only on $B_a$ 
it holds true that
\be
\nmm{G({\sf diag}(\a'), B', C') - G({\sf diag}(\a), B, C)}_2 \leq k_a\left[\nmm{B'-B}_F + \nmm{C'-C}_F\right].
\ee
\end{assumption}

Next define optimal regularizer similar to Equation \eqref{eq:opt_reg_P2}
\be\label{eq:opt_reg_P1}
\Omega^{\eqref{eq:P1}}(\hat{G}) := \inf_{n_x, \a, B, C \in \F_{\theta}^{\eqref{eq:P1}}}\sum_{j=1}^{n_x}\gamma(a_j)\nmm{\b_j}_2\nmm{\c_j}_2,\text{ s.t. }
G({\sf diag}(\a), B, C) = \hat{G}.
\ee
Now we present the statistical recovery guarantee of program \eqref{eq:P1}.
\begin{theoremA}\label{thm:thm_p1_stat}
Let $(U_i, Y_i)$ be $N$ i.i.d roll-outs following the Assumptions \ref{ass:data_model}.
Fix a $\delta \in (0, 1]$. 
Suppose the regularization parameter is such that
$\l \leq \tilde{\mathcal{O}}\left(\frac{n_x(n_y + n_u) + \ln(1/\delta)}{NL}\right)$.
For any global optimal points $(n_x, \{\hat{a}_j\}, \hat{B}, \hat{C})$ of 
program \eqref{eq:P1} satisfies Assumption \ref{ass:local_lip}. 
If $N/\ln(NL) \gtrsim {\sf cond}_{\F_{\theta}^{\eqref{eq:P1}}}^2\times[n_x(n_y + n_u) + \ln(1/\delta)]$,
then w.p at-least $1-\delta$ we have that
\vspace{-5pt}
\be\label{eq:stat_p1}
\hspace{-5pt}\nmm{(G(\diag(\{\hat{a}_j\}), \hat{B}, \hat{C}) - G^*){\tilde{\Sigma}_U}^{1/2}}^2_{F}\hspace{-2pt}
\leq \tilde{\mathcal{O}}\left(\left(\sigma^2  + \Omega^{\eqref{eq:P1}}(G^*)\right)\left[\frac{n_x(n_y + n_u) + \ln(1/\delta)}{NL}\right]\right).
\ee
\end{theoremA}
\begin{corollaryA}\label{crl:crl_p1_stat}
Under the conditions of Theorem \ref{thm:thm_p1_stat}, the following statements holds true,
\begin{enumerate}
    \item If $A^*$ is not real, diagonalizable then the upper bound of Equation \ref{eq:stat_p1} evaluates to $\infty$.
    \item Otherwise, the sample complexity for the recovery of $G^*$ with program \eqref{eq:P1} is
    $NL \gtrsim [n_x(n_y + n_u) + \ln(1/\delta)] \cdot {\sf polylog}(NL)$.
\end{enumerate}
\end{corollaryA}

\textbf{Remarks:} Statement 1 of Corollary \ref{crl:crl_p1_stat}
establishes that if the underlying system is not real, diagonalizable then
the statistical error obtained in Theorem \ref{thm:thm_p1_stat} becomes trivial.
Suppose that underlying system was real, diagonalizable then we would require
a total samples, $NL \gtrsim \tilde{\mathcal{O}}\left(n_x(n_y + n_x)\right)$. By 
na\"{\i}ve counting argument we have $n_x(n_y + n_u)$ parameters to estimate, therefore, we 
require
at-least those many samples to estimate all the parameters. This fact is indeed
reflected in Statement 2. In comparison
to Program \eqref{eq:P2} the statistical error rate and sample complexity
are $L$-folds tighter.

\textbf{Comparison with existing works.} Table~\ref{tab:comparision_stat} summarizes the comparison between our bounds and those from existing
key works. We observe that program~\eqref{eq:P2} achieves statistical error rates and sample complexity 
comparable to those of~\cite{lee2022improved}. 
However, unlike our approach, the method 
in~\cite{lee2022improved} assumes knowledge of 
the true system order, which we do not require.
For real, diagonalizable systems, program~\eqref{eq:P1} 
additionally removes the dependency on trajectory length present in program~\eqref{eq:P2} and prior methods. In 
terms of both statistical error rate and sample complexity, we observe at least an 
$L$-fold improvement.

\begin{table}[!ht]
    \centering
    \renewcommand{\arraystretch}{1.5}
    \resizebox{\textwidth}{!}{
    \begin{tabular}{|c|c|c|}
    \hline
    \rowcolor{gray!20}
       \textbf{Methods for low-order SysID} & \textbf{(Error rate$)^2$} & \textbf{Sample complexity}, $N_{\mathsf{tol}} \gtrsim$ \\
    \hline \hline
       Hankel nuclear norm minimization \citep{sun-et-al-ojcs22}  & ${{L^2n_x(n_y+n_u)}/{N_{\mathsf{tol}}}}$ & $L^2n_x(n_y+n_u)$ \\ 
       \hline
       SVD truncation \citep{lee2022improved} (\textit{$n_x^*$ is required})  & ${{Ln_x(n_y+n_u)}/{N_{\mathsf{tol}}}}$ & $Ln_x(n_y+n_u)$ \\  
       \hline
       BM re-parametrization \eqref{eq:P2} (Theorem \ref{thm:thm_p2_stat})  & ${{Ln_x(n_y+n_u)}/{N_{\mathsf{tol}}}}$ & $Ln_x(n_y+n_u)$ \\
       \hline
       System Parameters \eqref{eq:P1} (Theorem \ref{thm:thm_p1_stat}$)^{\ddagger}$  & ${{n_x(n_y+n_u)}/{N_{\mathsf{tol}}}}$ & $n_x(n_y+n_u)$ \\  \hline
    \end{tabular}
    }
    \vspace{-10pt}
    \caption{Summary of error rates and sample complexities for low-order SysID (up to log factors). Here, $N_{\mathbf{tol}} = N \cdot L$, and $\ddagger$ indicates restriction to real, diagonalizable systems.}
    \label{tab:comparision_stat}
\end{table}
\textbf{Proof Strategy.}
Although Programs \eqref{eq:P2} and \eqref{eq:P1} are nonconvex in the parameter space, the input-output map
remains linear. Thus, we can directly apply the time-dependent excess risk bounds from Theorem 6.1 of 
\cite{ziemann-tu-nips22}, under the additional condition that the impulse response is Lipschitz continuous with 
respect to the parameters and the parameter space is compact. These conditions are ensured by Assumptions 
\ref{ass:local_lip_hankel} and \ref{ass:local_lip}.

\section{Numerical Simulations}\label{sec:numerical_simulations}
\begin{figure}[!ht]
    \centering
    \subfigure[Recovery error]{
        \includegraphics[width=0.4\textwidth]{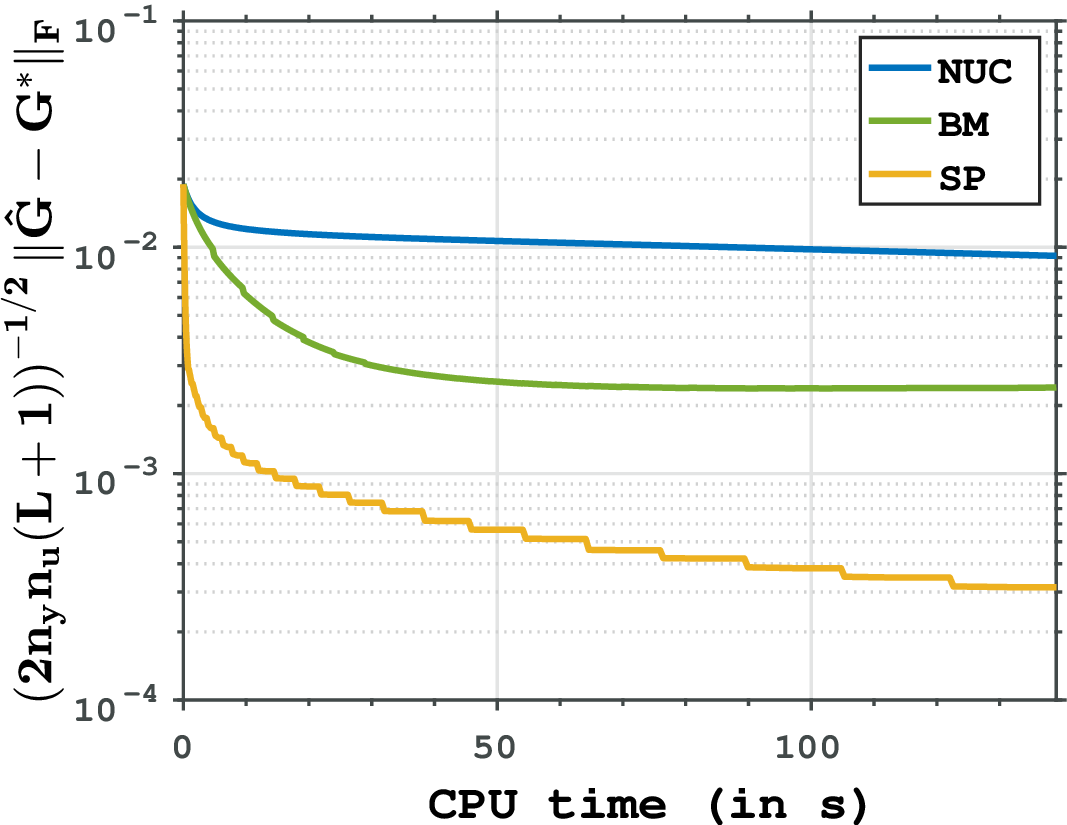}
        \label{fig:fig1}
    }
    \subfigure[Loss heuristics]{
        \includegraphics[width=0.4\textwidth]{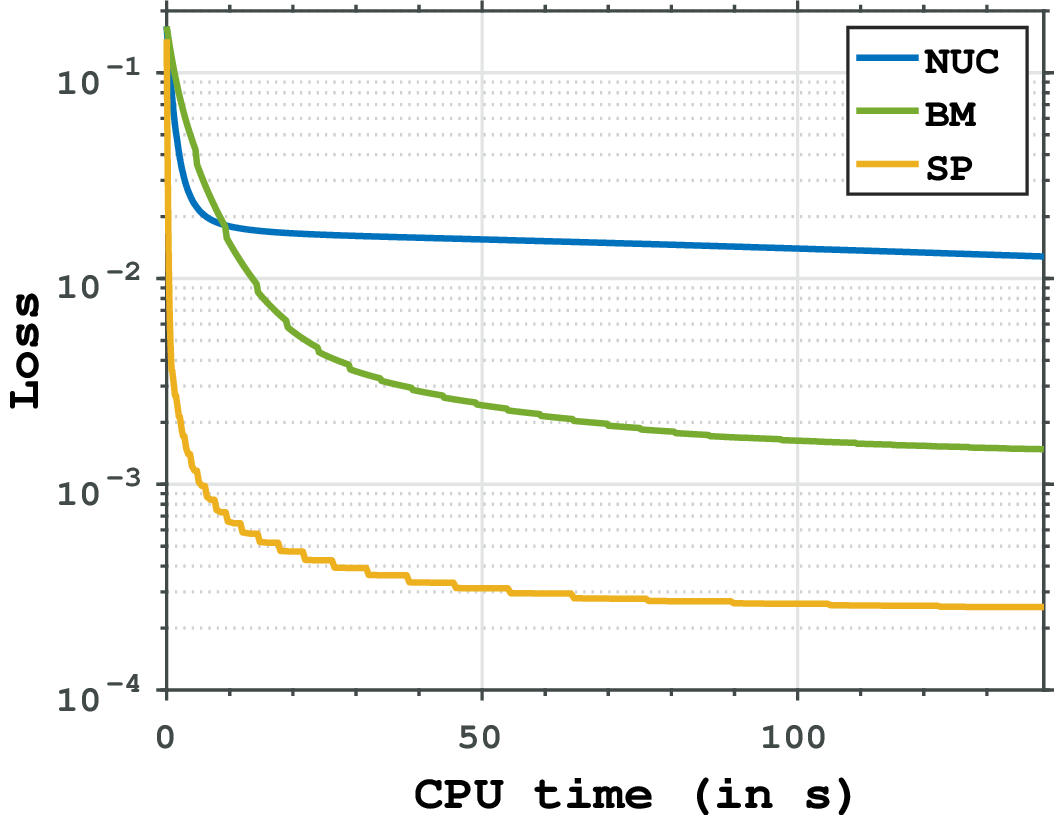}
        \label{fig:fig2}
    }
    \subfigure[Spectrum of Hankel matrix]{
        \includegraphics[width=0.4\textwidth]{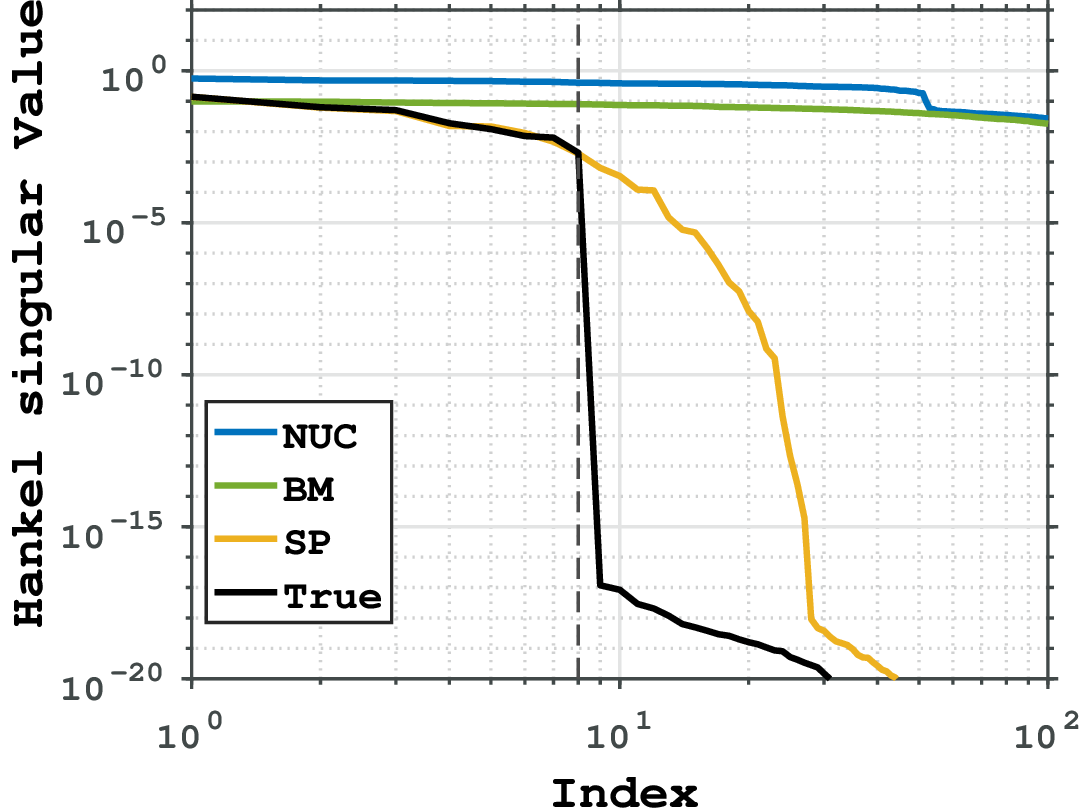}
        \label{fig:fig3}
    }
    \subfigure[{$\textsf{Polar}^{\eqref{eq:P2}/\eqref{eq:P1}}$}]{
        \includegraphics[width=0.4\textwidth]{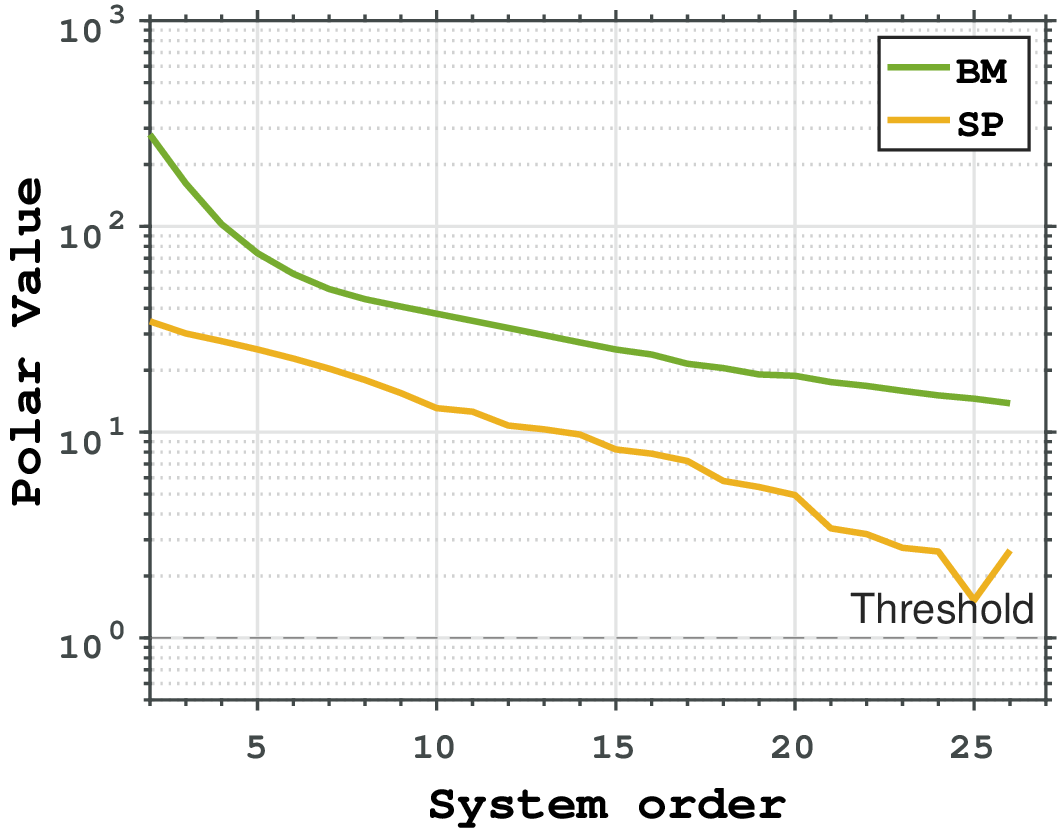}
        \label{fig:fig4}
    }

    \caption{Performance metrics for a real, diagnosable system. \textsf{NUC}, \textsf{BM},
    and \textsf{SP} corresponds to programs \eqref{eq:P1_1}, \eqref{eq:P2}, and \eqref{eq:P1}, respectively. Dashed
    line in \ref{fig:fig3} represents $n_x = 5$.}
    \label{fig:diag_metric}
\end{figure}
In this section, we compare low-order SysID using programs \eqref{eq:P1_1}, \eqref{eq:P2}, 
and \eqref{eq:P1} through numerical simulations, evaluating time complexity, sample 
complexity, and trajectory complexity.

\textbf{Data Generation.} We simulate linear trajectories
with system order $n_x^* = 5$ and dimensions $n_u = n_y = 8$. 
The system matrix $A$ is symmetric with Normal entries, 
while $B$ and $C$ are Normal matrices. Control 
inputs are sampled from $\mathcal{N}(0, I_{n_u}/n_u)$, and 
outputs are corrupted with Gaussian noise of variance $0.01$. 
We generate $N = 500$ rollouts, 
and trajectory length $2(L+1) = 102$.

\textbf{Algorithmic Implementation.} 
We use accelerated proximal gradient descent \eqref{alg:P1_1}\citep{becker-et-al-mpc11}
to solve convex program \eqref{eq:P1_1}. For programs 
\eqref{eq:P2}, and \eqref{eq:P1} we unroll Proposition \ref{prop:prop_p2}, 
and Theorem \ref{thm:thm_p1_opt} into algorithms \eqref{alg:P2}, \eqref{alg:P1} that uses Polyak's
gradient descent \citep{polyak-ussr64}. For each of the algorithm
we fixed regularization parameter, $\l = 0.001$
and choose the best learning rate, and momentum rate.
For fair comparison, recovery error, $\nmF{\hat{G}-G^*}/\sqrt{2n_yn_u(L+1)}$ was 
evaluated against CPU runtime rather than iterations of the algorithms. 
All algorithms are initialized such that
they share identical impulse response to ensure
consistent starting points.

\textbf{Ablation Studies.} We evaluate recovery loss across varying sample sizes,
trajectory lengths, and algorithm performance on real- and non-real-diagonalizable systems. 
From the results in Figures \ref{fig:diag_metric}, \ref{fig:vary_sample_trajectory_length}, 
and \ref{fig:non_diag_metric}, we conclude the following:
\begin{itemize}
    \item \textit{Computational Efficiency.} As shown in Figure \ref{fig:fig1}, program
    \eqref{eq:P1} outperforms program \eqref{eq:P2} and program 
    \eqref{eq:P1_1} in recovering the impulse response within a fixed CPU time budget. Figure 
    \ref{fig:fig2} highlights faster loss reduction for program \eqref{eq:P1}. Noise often 
    necessitates higher-order approximations for program \eqref{eq:P2} and 
    \eqref{alg:P1} to achieve global optima, as shown in Figure \ref{fig:fig3} and \ref{fig:fig4}.
    \item \textit{Statistical Efficiency.} For a fixed sample size or trajectory length 
    (Figure \ref{fig:vary_sample_trajectory_length}), program \eqref{eq:P1}
    consistently outperforms the others as discussed earlier in the Table \ref{tab:comparision_stat}.
\end{itemize}

\section{Conclusions}\label{sec:conclusion}
In this work, we focus on performing system identification with minimal state representation.
We propose two nonconvex reformulations: (i) a BM-style reparameterization of Hankel nuclear norm 
minimization, and (ii) direct optimization over the system parameters. Despite their nonconvex nature, we design 
algorithms with theoretical guarantees of 
global convergence using first-order optimization methods.
Additionally, we derive near tight statistical error rates and sample complexity bounds. Our findings 
suggest that directly estimating system parameters outperforms alternative approaches in both optimization 
efficiency and statistical recovery. However, our analysis is currently restricted to systems that are real,
diagonalizable. Extending this framework to handle non-diagonalizable systems is an avenue for future work.

\newpage
\acks{UKRT gratefully acknowledges Hancheng Min for his valuable feedback.
Other authors acknowledge the support of the Research Collaboration on the 
Mathematical and Scientific Foundations of Deep Learning 
(NSF grant 2031985 and Simons Foundation grant 814201).}

\bibliography{refs}

\newpage

\appendix

\counterwithin{theorem}{section}
\counterwithin{theoremA}{section}
\counterwithin{corollaryA}{section}
\counterwithin{problem}{section}
\counterwithin{assumption}{section}
\counterwithin{figure}{section}
\counterwithin{algorithm}{section}
\counterwithin{equation}{section}

\renewcommand{\thetheorem}{\Alph{section}.\arabic{theorem}}
\renewcommand{\theequation}{\Alph{section}.\arabic{equation}}
\renewcommand{\thetheoremA}{\Alph{section}.\arabic{theoremA}}
\renewcommand{\thecorollaryA}{\Alph{section}.\arabic{corollaryA}}
\renewcommand{\theproposition}{\Alph{section}.\arabic{proposition}}
\renewcommand{\theproblem}{\Alph{section}.\arabic{problem}}
\renewcommand{\theassumption}{\Alph{section}.\arabic{assumption}}
\renewcommand{\thefigure}{\Alph{section}.\arabic{figure}}
\renewcommand{\thealgorithm}{\Alph{section}.\arabic{algorithm}}

\section*{APPENDIX} \label{sec:appendix}

First in \textsection\ref{sec:apdx_opt}, we present the general framework of \cite{haeffele_global_2017} building upon this
later we discuss the proofs of Proposition \ref{prop:prop_p2},
Corollary \ref{crl:crl_p2}, and Theorem \ref{thm:thm_p1_opt} discussed in \textsection\ref{sec:opt_results}.
Then in \textsection\ref{sec:apdx_stat} we discuss relevant aspects of learning from \cite{ziemann-tu-nips22},
after-which we will discuss the proofs of Theorem \ref{thm:thm_p2_stat}, Corollary \ref{crl:crl_p2_stat}, 
Theorem \ref{thm:thm_p1_stat}, and Corollary \ref{crl:crl_p1_stat} that were discussed in
\textsection\ref{sec:stat_results}. 
Later, in \textsection\ref{sec:apdx_related} we discuss 
more related works in the area of low-order SysID.
Finally, in \textsection\ref{sec:apdx_numerical}
we present algorithms and more simulations that were discussed
in \textsection\ref{sec:numerical_simulations}.

\textbf{Extra notation.} 
$\langle f, g \rangle_q$ denotes the inner product between 
functions $f$ and $g$ with respect to the measure $q$, defined as 
$\int \langle f(x), g(x) \rangle \, dq(x)$. 
The notation ${\sf poly}(x)$ refers to 
$\mathcal{O}(x^{\alpha})$ for some $\alpha \geq 1$.

\section{Proofs for optimality certificates}\label{sec:apdx_opt}

Recall that we re-wrote the
Program \eqref{eq:P2}, \eqref{eq:P1} as the
predictions and regularization as
\be \label{eq:re_wrote}
\begin{split}
&{\Phi_{n_x}}^{\eqref{eq:P2}}(V, Z; U) := \sum_{j=1}^{n_x}\mathcal{H}^{\dagger}(\v_j\z_j^T)\vecc(U');
{\Theta_{n_x}}^{\eqref{eq:P2}}(V, Z) := \frac{1}{2}\sum_{j=1}^{n_x}\nmm{\v_j}_2^2 + \nmm{\z_j}_2^2, \\
&{\Phi_{n_x}}^{\eqref{eq:P1}}(\a, B, C; U)\hspace{-2pt}=\hspace{-2pt}\sum_{j=1}^{n_x}\left[P(a_j) \otimes \c_j\b_j^T\right]\vecc(U');
{\Theta_{n_x}}^{\eqref{eq:P1}}(\a, B, C)\hspace{-2pt}=\hspace{-2pt}\sum_{j=1}^{n_x}\gamma(a_j)\nmm{\b_j}_2\nmm{\c_j}_2.
\end{split}
\ee
\cite{haeffele_global_2017} provided global optimality
certificates for the objective of the form
\be\label{eq:gen}
\tag{GO}
\min_{r, \{W_j\}_{j=1}^{r}}\ell\left(Y, \sum_{j=1}^{r}\phi(W_j; X)\right) + \l \sum_{j=1}^{r}\theta(W_j),
\ee
where $\phi(\cdot)$, and $\theta(\cdot)$ are positively homogeneous with the same degree. More formally define polar for some measure $q$
\be\label{eq:polar}
\Omega_{q}^{\circ}(g) := \sup_{\theta(W) \leq 1}
\IP{g}{\phi(W)}_q
\ee
In practice, we solve Program~\eqref{eq:gen} for a 
fixed $r$ using first-order methods, 
such as gradient descent, Polyak's 
momentum~\citep{polyak-ussr64}, or 
Nesterov acceleration~\citep{nesterov-nag83}. 
Since the landscape is nonconvex, these methods 
typically converge to stationary points. 
Formally, convergence to a stationary point occurs only 
asymptotically; otherwise, the iterates 
remain within a bounded neighborhood of a stationary 
point. For simplicity of analysis, we assume access to a 
heuristic that allows exact convergence to a stationary 
point. Optimality guarantees for specific classes of 
objective functions were established 
by~\cite{haeffele_global_2017}, and we restate their 
result below.

\begin{theoremA}\label{thm:global}
For any stationary points $\{W_j\}$ which satisfies
the condition $\langle{-\frac{1}{\l}\nabla\ell(Y, \sum_{j=1}^{r}\phi(W_j))}$, ${\sum_{j=1}^{r}\phi(W_j)}\rangle_q$
= $\sum_{j=1}^{r}\theta(W_j)$. If  
$\Omega_{q}^{\circ}\left(-\frac{1}{\l}\nabla\ell(Y, \sum_{j=1}^{r}\phi(W_j))\right) = 1$, then
the program \eqref{eq:gen} is solved. Otherwise,
the objective can be reduced further by augmenting 
supremizer of the polar $W^*$ to $\{W_j\} \cup \tau W^*$, where $\tau$ is some scaling.
\end{theoremA}

Although, Theorem \ref{thm:global} provides
global optimality check and ability to progress
in the optimization. However, 
these guarantees hold true for the stationary
points that satisfy $\langle{-\frac{1}{\l}\nabla\ell(Y, \sum_{j=1}^{r}\phi(W_j))}$, ${\sum_{j=1}^{r}\phi(W_j)}\rangle_q$
= $\sum_{j=1}^{r}\theta(W_j)$. Next, we re-state 
result from \cite{haeffele_global_2017}
that proves that for a specific regularizer choices
every stationary point satisfies the earlier equality. 

\begin{proposition}[More general version of Proposition 3
\citep{haeffele_global_2017}]\label{prop:eql}\\
Given a function $\ell(Y, X)$ is convex in $X$ and once 
differentiable w.r.t. $X$; a constant $\l > 0$; and gauge 
functions ($\sigma_{w_1}(\w_1)$, $\dots$, $\sigma_{w_p}(\w_p)$), 
then for 
$\theta(W) = \sigma_{w_1}(\w_1)\times\dots\times\sigma_{w_p}(\w_p)$ or 
$\theta(W) = \frac{1}{p}\sum_{i=1}^{p}\left[\sigma^2_{w_i}(\w_i)\right]$,
any first-order optimal points $\{W_j\}$ of the program 
\eqref{eq:gen} satisfies the equality that
\be
\IP{-\frac{1}{\l}\nabla\ell(Y, \sum_{j=1}^{r}\phi(W_j))} {\sum_{j=1}^{r}\phi(W_j)}_q
= \sum_{j=1}^{r}\theta(W_j)
\ee
\end{proposition}

Now we are equipped with the background
to be able to analyze the global optimality 
guarantees for Programs \eqref{eq:P2}, and \eqref{eq:P1}.

\subsection{Proof of Proposition \ref{prop:prop_p2}}\label{sec:apdx_prop_p2}
Recall Proposition \ref{prop:prop_p2}.
\begin{proposition}
Let $(U_i, Y_i)$ be $N$ roll-outs of the system \eqref{eq:ls}. Consider the
estimator via factors $\hat{V}$ $\in$ $\R^{(L+1)n_y \times n_x}$, $\hat{Z}$ $\in$ $\R^{(L+1)n_u \times n_x}$. Define $U_i' := [U_i]_{1:2L+1}$, and 
\be
M := \frac{1}{N\l}\sum_{i=1}^{N}\left(\y_{2L+2}^{(i)} - \mathcal{H}^{\dagger}(\hat{V}\hat{Z}^T)\vecc(U_i')\right){\vecc(U_i')}^T;
\text{\textsf{Polar}}^{\eqref{eq:P2}}:=\nmm{\mathcal{H}(M)}_2.
\ee
Suppose $\hat{V}, \hat{Z}$ are any stationary points of Program \eqref{eq:P2}. 
If $\text{\textsf{Polar}}^{\eqref{eq:P2}} = 1$, then $\hat{V}$ and $\hat{Z}$ are globally optimal. 
Otherwise, the objective can be reduced 
by augmenting the top-singular vectors $(\v^*, \z^*)$
of $\mathcal{H}(M)$ to 
$\begin{bmatrix}\hat{V} & \tau^*\v^*\end{bmatrix}$ $\in$ $\mathbb{R}^{(L+1)n_y \times (n_x+1)}$ and 
$\begin{bmatrix}\hat{Z} & \tau^*\z^*\end{bmatrix}$ $\in$ 
$\mathbb{R}^{(L+1)n_u \times (n_x+1)}$, for some scaling factor $\tau^*$.
\end{proposition}
\begin{proof}
Recall that our program consists of the predictor
and the regularizer of the form,
\be 
{\Phi_{n_x}}^{\eqref{eq:P2}}(V, Z; U) := \sum_{j=1}^{n_x}\mathcal{H}^{\dagger}(\v_j\z_j^T)\vecc(U');
{\Theta_{n_x}}^{\eqref{eq:P2}}(V, Z) := \frac{1}{2}\sum_{j=1}^{n_x}\nmm{\v_j}_2^2 + \nmm{\z_j}_2^2.
\ee
In order to apply the Theorem \ref{thm:global} we need to

(i) check the equality 
$\langle{-\frac{1}{\l}\nabla\ell(Y, \sum_{j=1}^{r}\phi(W_j))}$, ${\sum_{j=1}^{r}\phi(W_j)}\rangle_q$
= $\sum_{j=1}^{r}\theta(W_j)$, and

(ii) compute the polar, 
$\Omega_{q}^{\circ}(\cdot)$.

First, we note that
${\Theta_{n_x}}^{\eqref{eq:P2}}$ is a sum of two gauge
functions then for any stationary points $\hat{V}, \hat{Z}$
from Proposition \ref{prop:eql} we have that,
\be
\IP{Y - \mathcal{H}^{\dagger}(\hat{V}\hat{Z}^T)\vecc(U')}
{\mathcal{H}^{\dagger}(\hat{V}\hat{Z}^T)\vecc(U')}_{\mu_N} =
\frac{\l}{2}\left[\nmF{\hat{V}}^2 + \nmF{\hat{Z}}^2\right],
\ee
here $\mu_N$ being the empirical measure with $N$-samples.\\
Second, we compute the polar for the empirical measure, $\mu_N$,
from the definition in Equation \eqref{eq:polar},
\be
\Omega^{\circ}_{\mu_N}(Z) = 
\sup_{\nmm{\v}_2^2 + \nmm{\z}_2^2 \leq 2} \IP{Z}{\mathcal{H}^{\dagger}(\v\z^T)\vecc(U')}_{\mu_N},
\ee
supremum over the set $\{\nmm{\v}_2^2 + \nmm{\z}_2^2 \leq 2\}$,
and $\{\nmm{\v}_2 \leq 1, \nmm{\z}_2 \leq 1\}$
by the homogeneity, then we have,
\be
\Omega^{\circ}_{\mu_N}(Z) = 
\sup_{\nmm{\v}_2 \leq 1, \nmm{\z}_2 \leq 1} \IP{Z}{\mathcal{H}^{\dagger}(\v\z^T)\vecc(U')}_{\mu_N},
\ee
by moving $\vecc(U')$ to the other side of the product we obtain,
\be
\Omega^{\circ}_{\mu_N}(Z) = 
\sup_{\nmm{\v}_2 \leq 1, \nmm{\z}_2 \leq 1} \IP{Z\vecc(U')^T}{\mathcal{H}^{\dagger}(\v\z^T)}_{\mu_N},
\ee
now we apply the adjoint of the $\mathcal{H}^{\dagger}$ that
leads us to
\be
\Omega^{\circ}_{\mu_N}(Z) = 
\sup_{\nmm{\v}_2 \leq 1, \nmm{\z}_2 \leq 1} \IP{\mathcal{H}\left(Z\vecc(U')^T\right)}{\v\z^T}_{\mu_N}
= \nmm{\frac{1}{N}\sum_{j=1}^{N}\mathcal{H}\left(Z_i\vecc(U'_i)^T\right)}_2.
\ee
Now we apply Theorem \ref{thm:global} for 
$Z = \frac{1}{\l}(\y_{2L+2} - \mathcal{H}^{\dagger}(\hat{V}\hat{Z}^T)\vecc(U'))$ which concludes 
our results.
\end{proof}

\subsection{Proof of Corollary \ref{crl:crl_p2}}\label{sec:apdx_crl_p2}
Recall the Corollary \ref{crl:crl_p2}.
\begin{corollaryA}
Under the conditions of Proposition \ref{prop:prop_p2}, if $(\hat{n}_x, \hat{V}, \hat{Z})$ are the globally
optimal points of program \eqref{eq:P2}, then the system has the order 
$\hat{n}_x$ and system parameters take the form
\be
\hat{A} = \left(\left[\hat{V}\right]_{1:n_y, :}\right)^{\dagger}\left[\hat{V}\right]_{1+n_y:2n_y, :},
\hat{B} = \left[\hat{Z}^T\right]_{:, 1:n_u},
\hat{C} = \left[\hat{V}\right]_{1:n_y, :}.
\ee
\end{corollaryA}
\begin{proof}
As $\hat{V}, \hat{Z}$ are the global optimal points of the program
\eqref{eq:P2} whose minima matches with nuclear norm minimization
problem. Suppose that $\mathcal{H}(H^*) = L\Sigma R^T$, where
$H^*$ is the optima of Program \eqref{eq:P1_1}, then we have 
that for some orthonormal matrix, $O \in \R^{\hat{n}_x \times \hat{n}_x}$,
\be
\hat{V} = L\Sigma O, \hat{Z} = R \Sigma ^T O.
\ee
Now implicitly we have computed the SVD of the Hankel matrix
now we use $\hat{V}$ as observability matrix, i.e.,
\be
\hat{V} = L\Sigma O = \begin{bmatrix}
\hat{C} \\ \hat{C}\hat{A} \\ \vdots \\ \hat{C}\hat{A}^{L}
\end{bmatrix},
\ee
and similarly $\hat{Z}^T$ as the controllability matrix,
\be
\hat{Z}^T = \begin{bmatrix}\hat{B} & \hat{A}\hat{B} & \dots & \hat{A}^L\hat{B}\end{bmatrix}.
\ee
Now $(\hat{A}, \hat{B}, \hat{C})$ can be recovered by the
least-squares solution concluding our proof.
\end{proof}

\subsection{Proof of Theorem \ref{thm:thm_p1_opt}}\label{sec:apdx_thm_p1_out}

Now we apply the similar techniques for the Program
\eqref{eq:P1}. However, one main technical difficulty in
doing so is the fact the predictor is not
positively homogeneous with respective to parameter $A$.
We cannot directly apply Theorem \ref{thm:global}.
To circumvent this, \cite{tadipatri-et-al-iclr25} 
relaxed the assumption strict positive homogeneity 
to a weaker notion described as the following assumption.

\begin{assumption}[Balanced Homogeneity of $\phi$ and $\theta$]\label{ass:bal}
The factor map $\phi$ and the regularization map $\theta$ can be scaled equally by non-negative scaling of (a subset of) the parameters. 
Formally, we assume that there exists sub-parameter spaces
$(\mathcal{K}, \H)$ from the parameter space $\mathcal{W}$
such that $\mathcal{K} \times \H = \mathcal{W}$, $\forall ({\mathbf{k}}, {\mathbf{h}}) \in (\mathcal{K}, \H)$, and $\beta \geq 0$ we have 
$\phi((\beta {\mathbf{k}}, {\mathbf{h}})) = \beta^{p}\phi(({\mathbf{k}}, {\mathbf{h}}))$ and 
$\theta((\beta {\mathbf{k}}, {\mathbf{h}})) = \beta^{p}\theta(({\mathbf{k}}, {\mathbf{h}}))$ for some $p > 0$.
Further, we assume that for bounded input $X$ the set $\{ \phi(W)(X) : \forall W \in \mathcal{W}\text{ s.t. }\theta(W) \leq 1 \}$ is bounded.
\end{assumption}

For this weaker notion of homogeneity it was proven that
results of Proposition \ref{prop:eql} still hold true. Formally,
\begin{proposition}\label{prop:stat_points}
Under Assumption \ref{ass:bal}, if $\{W_j\}$ are stationary points of the Program \eqref{eq:gen}, then for all $j \in [r],$
\be
\forall j \in [r]; \IP{-\frac{1}{\l}\nabla \ell(Y, \sum_{j=1}^{r}\phi(W_j))}{\phi(W_j)}_{\mu_N} = \theta(W_j).
\ee
\end{proposition}
Recall the Theorem \ref{thm:thm_p1_opt}.
\begin{theoremA}
Let $(U_i, Y_i)$ be $N$ roll-outs from the system \eqref{eq:ls} with $\x_0 = 0$. 
Consider the
estimator of system parameters via $\big({\textsf{diag}}\left(\a\right)$, $\begin{bmatrix}\b_1, \b_2, \dots, \b_{n_x}\end{bmatrix}^T, \begin{bmatrix}\c_1, \c_2, \dots, \c_{n_x}\end{bmatrix}\big)$, of order $n_x$ where $\a \in \R^{n_x}$, $\b_{\cdot} \in \R^{n_u}$, $\c_{\cdot} \in \R^{n_y}$. 
Define $U_i' := [U_i]_{1:2L+1}$, and 
\be
M(a)\hspace{-2pt}:= \frac{1}{2N(L+1)\l}\sum_{i=1}^{N}\left[\hspace{-1pt}Y_i - \sum_{j=1}^{n_x}\c_j\b_j^TU_i'P^T(a_j)\hspace{-2pt}\right]\hspace{-1pt}\frac{P(a)}{\gamma(a)}{U_i'}^{T}\text{, }
\text{\textsf{Polar}}^{\eqref{eq:P1}}\hspace{-2pt}:= \sup_{a \in {\R}} \nmm{M(a)}_2.
\ee
Let $\{a_j, \c_j, \b_j\}_{j=1}^{n_x}$ be any stationary points of program \eqref{eq:P1}.
If $\text{\sffamily Polar}^{\eqref{eq:P1}} = 1$, then $\{a_j, \c_j, \b_j\}$ are global optimal points. Otherwise, we can
reduce the objective by the parameters 
$\{a_j, \c_j, \b_j\}_{j=1}^{n_x}$ $\cup$ $\{\tau_1a^*$ $,\tau\c^*,$ $\tau\b^*\}$ for
some scalings, $\tau_1, \tau$, where
$a^*$ is the supremizer of $\kappa$ and $\c^*, \b^*$ are the top-singular vectors the matrix $M(\a^*)$.
\end{theoremA}

\begin{proof}
Recall that our program consists of the predictor
and the regularizer of the form,
\be 
{\Phi_{n_x}}^{\eqref{eq:P1}}(\a, B, C; U)\hspace{-2pt}=\hspace{-2pt}\sum_{j=1}^{n_x}\left[P(a_j) \otimes \c_j\b_j^T\right]\vecc(U);
{\Theta_{n_x}}^{\eqref{eq:P1}}(\a, B, C)\hspace{-2pt}=\hspace{-2pt}\sum_{j=1}^{n_x}\gamma(a_j)\nmm{\b_j}_2\nmm{\c_j}_2.
\ee
For any stationary points, $\{a_j, \c_j, \b_j\}$ from
Proposition \ref{prop:stat_points} we have
\be
\IP{(\vecc(Y) - {\Phi_{n_x}}^{\eqref{eq:P1}}(\a, B, C; U))}{{\Phi_{n_x}}^{\eqref{eq:P1}}(\a, B, C; U)}_{\mu_N}
= 2\l (L+1) {\Theta_{n_x}}^{\eqref{eq:P1}}(\a, B, C).
\ee
The rest of the proof is just computing the polar
for this program \eqref{eq:P1}. Recall the definition,
\be
\Omega_{\mu_N}^{\circ}(\z) = \sup_{\gamma(a)\nmm{\c}_2\nmm{\b}_2 \leq 1}\IP{\z}{\left[P(a) \otimes \c\b^T \right] \vecc{(U)}}_{\mu_N}.
\ee
By de-vectorizing we have,
\be
\Omega_{\mu_N}^{\circ}(\z) = \sup_{\gamma(a)\nmm{\c}_2\nmm{\b}_2 \leq 1}\IP{\vecc^{\dagger}(\z)}{\vecc^{\dagger}\left(\left[P(a) \otimes \c\b^T \right] \vecc{(U)}\right)}_{\mu_N}.
\ee
Now use the identity $(C \otimes B)\vecc(X) = \vecc(BXC^T)$, then we
obtain,
\be
\Omega_{\mu_N}^{\circ}(\z) = \sup_{\gamma(a)\nmm{\c}_2\nmm{\b}_2 \leq 1}\IP{\vecc^{\dagger}(\z)}{\c\b^TUP^T(a)}_{\mu_N}.
\ee
Now by rearranging the terms across the inner product we obtain,
\be
\Omega_{\mu_N}^{\circ}(\z) = \sup_{\gamma(a)\nmm{\c}_2\nmm{\b}_2 \leq 1}\IP{\vecc^{\dagger}(\z)P(a)U^T}{\c\b^T}_{\mu_N}
= \sup_{a \in \R}\nmm{\frac{1}{N}\sum_{i=1}^{N}\vecc^{\dagger}(\z_i)\frac{P(a)}{\gamma(a)}U_i^T}_2.
\ee
Now set $\z = \frac{1}{2\l(L+1)}\vecc(Y) - \sum_{j=1}^{n_x}\left[P(a_j) \otimes \c_j\b_j^T\right]\vecc(U)$ that is same as
\be
\vecc^{\dagger}(\z) = \frac{1}{2\l(L+1)}\left[Y - \sum_{j=1}^{n_x}\c_j\b_j^TUP^T(a_j)\right].
\ee
Now plugging the polar in Theorem \eqref{thm:global} concludes the
proof.
\end{proof}

\section{Proofs for statistical bounds}\label{sec:apdx_stat}
First we present Theorem \ref{thm:apdx_timedep_risk} that
provides excess risk or recovery error bound
for time dependent data with compact hypothesis class.
These results take ideas from 
\cite{ziemann-tu-nips22, ziemann-et-al-icml24}. While
their analysis was catered to more general hypothesis class.
However, our settings is applied to simple linear hypothesis
classes.

Now we state our assumptions on the time-dependent data 
and the hypothesis class.
\begin{assumption}[Data assumption]\label{ass:data}
The sequence $\{\x_t\}_{t \in [T]}$ are random variables such that 
there exists a $r_1 \geq 0$ such that for any $r \geq r_1$ we have
that
\be
\mathbb{P}\left(\{\x_t\}_{t \in [T]} \notin \mathbb{B}(r)\right)
\leq \mathcal{O}\left(\exp\left(-r^2/\sigma^2\right)\right).
\ee
\end{assumption}
This assumption trivially holds if $\x_t$ is a sub-Gaussian random 
variable.
\begin{assumption}[Hypothesis class]\label{ass:hypo}
The parametric class $\F_{\theta}$, and function $f_{\cdot}: \R^{n_x} \to \R^{n_y}$ are such that for any fixed function $f^*: \R^{n_x} \to \R^{n_y}$ and a semi-metric $d$
there exists a $r_2 \geq 0$ such that for any $r \geq r_2$
the following holds true
\begin{enumerate}
    \item $\forall \theta, \theta' \in \F_{\theta}$ we have that
    $\sup_{\z \in \mathbb{B}(r)}\nmm{f_{\theta}(\z)-f_{\theta'}(\z)}_2 \leq {\sf poly}(r)d(\theta, \theta')$,
    \item $\sup_{\z \in \mathbb{B}(r), \theta \in \F_{\theta}}
    \nmm{f_{\theta}(\z) - f^*(\z)} = {\sf poly}(r)$,
    \item $\N(\F_{\theta}, d, \epsilon) \leq \mathcal{O}(1/\epsilon^{{\sf dim}(\F_{\theta})})$.
\end{enumerate}
\end{assumption}
If the function $f_{\theta}$ is linear in its inputs, then
Assumption~\ref{ass:hypo} holds, which is our primary concern.\\
Define `condition' number for the hypothesis class $\F_{\theta}$, 
and predictor $f_{\theta}: \R^{n_x} \to \R^{n_y}$; 
$\forall \theta \in \F_{\theta}$ as
\be\label{eq:cond}
{\mathsf cond}_{\F_{\theta}} := \max_{\theta \in \F_{\theta}}
\max_{t \in [T]}\frac{\sqrt{\mathbb{E}\left[\nmm{f_{\theta}(\x_t)}_2^4\right]}}{\mathbb{E}\left[\nmm{f_{\theta}(\x_t)}_2^2\right]}.
\ee
Now we state recovery error bounds for time dependent data.
\begin{theoremA}\label{thm:apdx_timedep_risk}
Let $\y_t = f^*(\x_t) + \v_t$ where $\x_t$ satisfies Assumption~\ref{ass:data}, and 
$\v_t \,|\, \x_{1:t}$ are conditionally independent 
$\sigma^2$-sub-Gaussian vectors in $\mathbb{R}^{n_y}$. 
Assume that the predictor function $f_{\theta}$, for all $\theta \in \F_{\theta}$, satisfies Assumption~\ref{ass:hypo}. 
Consider the parametric estimator $\hat{\theta}$, 
based on access to $N$ i.i.d. rollouts 
$\{(\x_{t}^{i}, \y_{t}^{i})\}$, obtained by solving the program
\begin{equation}\label{eq:emp_obj}
\hat{\theta} \in \argmin_{\theta \in \F_{\theta}} \frac{1}{NT} \sum_{i=1}^{N} \sum_{t=1}^{T} 
\nmm{\y_t^i - f_{\theta}(\x_t^i)}_2^2 + \lambda \Omega(f_{\theta}),
\end{equation}
where $\Omega: \F \to \mathbb{R}_+$ is a regularization function and $\lambda \in \mathbb{R}_+$.

Then, for any fixed $\delta \in (0,1]$, if
\[
\frac{N}{\ln(NT)} \gtrsim {\sf cond}_{\F_{\theta}}^2 \times 
\left[{\sf dim}(\F_{\theta}) + \ln\left(\frac{1}{\delta}\right)\right],
\]
w.p at least $1-\delta$ we have that
\begin{equation}\label{eq:excess_risk}
\nmm{f_{\hat{\theta}} - f^*}_{L_2}^2
\leq \tilde{\mathcal{O}}\left(\sigma^2 \frac{{\sf dim}(\F_{\theta}) + \ln(1/\delta)}{NT}\right) + \lambda \Omega(f^*).
\end{equation}
\end{theoremA}

We will defer the proof to \textsection\ref{sec:apdx_timedep_risk_proof}.
The ideas for the proof of Theorem \ref{thm:apdx_timedep_risk} are
directly inspired from
\cite{ziemann-tu-nips22, ziemann-et-al-icml24}.
Although, their bounds were restricted to un-regularized problem in this work we extend these proofs
to regularized versions in \textsection\ref{sec:apdx_timedep_risk_proof}. 
Observe that due to regularization term objective \eqref{eq:emp_obj},
the obtained solution will be bias. This bias can be observed in the bound of Equation \eqref{eq:excess_risk}. Therefore, if we choose $\l$ to decay at the appropriate rate then we can
achieve bounds that decay with more samples.\\
Now we state a Corollary that chooses appropriate $\l$.
\begin{corollaryA}\label{crl:excess_risk_recovery}
Under the conditions of Theorem \ref{thm:apdx_timedep_risk}, suppose $\l \leq \tilde{\mathcal{O}}\left(\frac{{\sf dim}(\F_{\theta}) + \ln(1/\delta)}{NT}\right)$ then we have that
\be\label{eq:excess_recovery}
\mathbb{P}\left(\nmm{f_{\hat{\theta}} - f^*}_{L_2}^2 \leq \tilde{\mathcal{O}}\left({\left[\sigma^2 + \Omega(f^*)\right]\frac{{\sf dim}(\F_{\theta}) + \ln(1/\delta)}{NT}}\right)\right) \geq 1 - \delta.
\ee
\end{corollaryA}
This corollary is straight a direct consequence of Theorem \ref{thm:apdx_timedep_risk}. As $N\to \infty$, the upper bound
of Equation \ref{eq:excess_recovery} goes to 0 thereby achieving 
recovery. Next we apply Theorem \ref{thm:apdx_timedep_risk}
to Programs \eqref{eq:P2}, and \eqref{eq:P1}.

\subsection{Proof of Theorem \ref{thm:thm_p2_stat}}\label{sec:apdx_thm_p2_stat}
Recall the statistical recovery of 
the program \eqref{eq:P2}.

\begin{theoremA}\label{thm:apdx_thm_p2_stat}
Let $(U_i, Y_i)$ be $N$ i.i.d roll-outs following the Assumptions \ref{ass:data_model}.
Fix a $\delta \in (0, 1]$. 
Suppose the regularization parameter is such that
$\l \leq \tilde{\mathcal{O}}\left(\frac{n_x(n_y + n_u)}{N} + \frac{\ln(1/\delta)}{NL}\right)$.
For any global optimal points $(n_x, \hat{V}, \hat{Z})$ of program \eqref{eq:P2} satisfying Assumption \ref{ass:local_lip_hankel}.
If $N/\ln(NL) \gtrsim {\sf cond}_{\F_{\theta}^{\eqref{eq:P2}}}^2 \times [Ln_x(n_y + n_u) + \ln(1/\delta)]$, then
w.p at-least $1-\delta$ we have that
\be\label{eq:apdx_stat_p2}
\begin{split}
\nmm{(G'(\hat{V}, \hat{Z})-G^*){\tilde{\Sigma}_{U}}^{1/2}}_{F}^2
\leq \tilde{\mathcal{O}}\left(\left(\sigma^2  + \Omega^{\eqref{eq:P2}}(G^*)\right)\left[\frac{n_x(n_y + n_u)}{N} + \frac{\ln(1/\delta)}{NL}\right]\right).
\end{split}
\ee
\end{theoremA}

\begin{proof}
This result is straight forward application of Corollary 
\ref{crl:excess_risk_recovery} (which in-turn is a direct
consequence of Theorem \ref{thm:apdx_timedep_risk}).
First let us convert \eqref{eq:ls} to the problem setting
of Theorem \ref{thm:apdx_timedep_risk}.
Define $E_t \in \R^{n_y \times 2(L+1)n_y}$ where only
$[E_t]_{1 + (t-1)n_y: tn_y} = I_{n_y}$.
\be
\y_t = E_tG^*\vecc(U) + \zeta_t 
     = \left[\vecc(U)^T \otimes E_t\right]\vecc(G^*) + \zeta_t.
\ee

Define $\tilde{U}_t := \vecc(U)^T \otimes E_t \in \R^{n_y \times 4(L+1)^2n_yn_u}$, then
we have
\be
\y_t = \tilde{U}_t\vecc(G^*) + \zeta_t.
\ee

This concludes the algebraic manipulation to resemble
the problem setting in Theorem \ref{thm:apdx_timedep_risk}.
We now instantiate Corollary \ref{crl:excess_risk_recovery}
with parametric class $\F_{\theta}^{\eqref{eq:P2}}$,
$f^*(X_t) = \tilde{U}_t\vecc(G^*)$, and $f_{\theta}(X_t)
= \tilde{U}_t\vecc(G'(V, Z))$, where $G'(V, Z)$ is 
impulse response matrix of the system with 
Markov parameters $\H^{\dagger}(VZ^T)$.

\be
\begin{split}
{\sf cond}_{\F_{\theta}^{\eqref{eq:P2}}}
= \max_{(V, Z) \in \F_{\theta}^{\eqref{eq:P2}}}
\max_{t \in [2(L+1)]}\frac{\sqrt{\mathbb{E}\left[\nmm{ \tilde{U}_t\vecc(G'(V, Z)-G^*)}_2^4\right]}}{\mathbb{E}\left[\nmm{\tilde{U}_t\vecc(G'(V, Z)-G^*)}_2^2\right]}
\end{split}
\ee
The condition number
${\sf cond}_{\F_{\theta}^{\eqref{eq:P2}}}$ is constant
if $U$ follows Gaussian distribution. For the cases,
when $U$ following sub-Gaussian we can upper bound this
object through Proposition 6.1 from 
\cite{ziemann_tutorial_2024}. Finally,
we have that ${\sf dim}(\F_{\theta}^{\eqref{eq:P2}})$
is the ambient dimension of $(V, Z)$ that amounts to
$(L+1)n_x(n_u + n_y)$. 
With these properties applied to Corollary 
\ref{crl:excess_risk_recovery} we obtain,
\be
\begin{split}
\mathbb{P}\Big(&\left(\frac{1}{T}\sum_{t=1}^{T}\mathbb{E}'\left[\nmm{\tilde{U_t}\vecc(G'(\hat{V}, \hat{Z})-G^*)}_2^2\right]\right)^{1/2} \\
&\leq \tilde{\mathcal{O}}\left(\sqrt{\left[\sigma^2 + \Omega^{\eqref{eq:P2}}(G^*)\right]\frac{(L+1)n_x(n_u + n_y) + \ln(1/\delta)}{N2(L+1)}}\right)\Big)\\
&\geq 1 - \delta.
\end{split}
\ee
We simplify the left side term to Hankel norm via linearity
of expectation we obtain
\be
\begin{split}
&\frac{1}{T}\sum_{t=1}^{T}\mathbb{E}'\left[\nmm{\tilde{U_t}\vecc(G'(\hat{V}, \hat{Z})-G^*)}_2^2\right]\\
&= \mathbb{E}\left[\IP{\vecc(G'(\hat{V}, \hat{Z})-G^*)\vecc(G'(\hat{V}, \hat{Z})-G^*)^T}{\frac{1}{T}\sum_{t=1}^{T}\tilde{U}_t^T\tilde{U}_t}\right].
\end{split}
\ee
From the definition of $\tilde{U}_t^T\tilde{U}_t = (\vecc(U) \otimes E_t^T)
(\vecc(U)^T \otimes E_t)$ = $(\vecc(U)\vecc(U)^T \otimes E_t^TE_t)$. Using this identity we have
\be
\begin{split}
&\frac{1}{T}\sum_{t=1}^{T}\mathbb{E}'\left[\nmm{\tilde{U_t}\vecc(G'(\hat{V}, \hat{Z})-G^*)}_2^2\right]\\
&= \mathbb{E}'\left[\IP{\vecc(G'(\hat{V}, \hat{Z})-G^*)\vecc(G'(\hat{V}, \hat{Z})-G^*)^T}{\vecc(U)\vecc(U)^T \otimes \frac{1}{T}\sum_{t=1}^{T}E_t^TE_t}\right].
\end{split}
\ee
We have that 
$\frac{1}{T}\sum_{t=1}^{T}E_t^TE_t = I_{2(L+1)n_y}$ this simplifies our object of interest to
\be
\begin{split}
&\frac{1}{T}\sum_{t=1}^{T}\mathbb{E}'\left[\nmm{\tilde{U_t}\vecc(G'(\hat{V}, \hat{Z})-G^*)}_2^2\right]\\
&= \mathbb{E}'\left[\IP{\vecc(G'(\hat{V}, \hat{Z})-G^*)\vecc(G'(\hat{V}, \hat{Z})-G^*)^T}{\vecc(U)\vecc(U)^T \otimes I_{2(L+1)n_y}}\right].
\end{split}
\ee
Now we move the expectation inside and use the property 
that $\mathbb{E}'[\vecc(U)\vecc(U)^T]$ = $I_{2(L+1)} \otimes \Sigma_U$ to obtain
\be
\begin{split}
&\frac{1}{T}\sum_{t=1}^{T}\mathbb{E}'\left[\nmm{\tilde{U_t}\vecc(G'(\hat{V}, \hat{Z})-G^*)}_2^2\right]\\
&= \IP{\vecc(G'(\hat{V}, \hat{Z})-G^*)\vecc(G'(\hat{V}, \hat{Z})-G^*)^T}{(I_{2(L+1)} \otimes \Sigma_{U}) \otimes I_{2(L+1)n_y}}.
\end{split}
\ee

We can factorize the right argument in the inner product 
to a quadratic form
\be
\begin{split}
&\frac{1}{T}\sum_{t=1}^{T}\mathbb{E}'\left[\nmm{\tilde{U_t}\vecc(G'(\hat{V}, \hat{Z})-G^*)}_2^2\right]\\
= &\langle{\vecc(G'(\hat{V}, \hat{Z})-G^*)\vecc(G'(\hat{V}, \hat{Z})-G^*)^T},\\
&{((I_{2(L+1)} \otimes \Sigma_{U}^{1/2}) \otimes I_{2(L+1)n_y})((I_{2(L+1)} \otimes \Sigma_{U}^{1/2}) \otimes I_{2(L+1)n_y})^T} \rangle.
\end{split}
\ee

Now re-arrange the terms in the inner product we obtain
\be
\begin{split}
&\frac{1}{T}\sum_{t=1}^{T}\mathbb{E}'\left[\nmm{\tilde{U_t}\vecc(G'(\hat{V}, \hat{Z})-G^*)}_2^2\right]\\
&= \mathbb{E}'\left[\nmm{((I_{2(L+1)} \otimes \Sigma_{U}^{1/2}) \otimes I_{2(L+1)n_y})\vecc(G'(\hat{V}, \hat{Z})-G^*)}_2^2\right].
\end{split}
\ee
Observe that $((I_{2(L+1)} \otimes \Sigma_{U}^{1/2}) \otimes I_{2(L+1)n_y})$ is the same as $I_{4(L+1)^2n_y} \otimes \Sigma_{U}^{1/2}$ using this property we obtain
\be
\begin{split}
&\frac{1}{T}\sum_{t=1}^{T}\mathbb{E}'\left[\nmm{\tilde{U_t}\vecc(G'(\hat{V}, \hat{Z})-G^*)}_2^2\right]\\
&= \nmm{(I_{4(L+1)^2n_y} \otimes \Sigma_{U}^{1/2})\vecc(G'(\hat{V}, \hat{Z})-G^*)}_2^2.
\end{split}
\ee
Now we use the Kronecker product property again to obtain
\be
\begin{split}
&\frac{1}{T}\sum_{t=1}^{T}\mathbb{E}'\left[\nmm{\tilde{U_t}\vecc(G'(\hat{V}, \hat{Z})-G^*)}_2^2\right]\\
&= \nmm{(G'(\hat{V}, \hat{Z})-G^*)^T}_{I_{2(L+1)} \otimes \Sigma_U}^2.
\end{split}
\ee
This concludes the proof.
\end{proof}

\subsection{Proof of Corollary \ref{crl:crl_p2_stat}}\label{sec:apdx_crl_p2_stat}
\begin{corollaryA}
Under the conditions of Theorem \ref{thm:thm_p2_stat},
the sample complexity for the recovery of $G^*$
through program \eqref{eq:P1} is
$NL \gtrsim (Ln_x(n_y + n_u) + \ln(1/\delta)) \cdot {\sf polylog}(NL)$.
\end{corollaryA}
\begin{proof}
From Equation \eqref{eq:apdx_stat_p2} it is sufficient
to have
\be
NL \gtrsim (Ln_x(n_y + n_u) + \ln(1/\delta)) \cdot {\sf polylog}(NL),
\ee
in-order for the right side term in Equation \eqref{eq:apdx_stat_p2} to decay asymptotically.
\end{proof}

\subsection{Proof of Theorem \ref{thm:thm_p1_stat}}\label{sec:apdx_thm_p1_stat}
Recall the statistical recovery of 
the program \eqref{eq:P1}.

\begin{theoremA}
Let $(U_i, Y_i)$ be $N$ i.i.d roll-outs following the Assumptions \ref{ass:data_model}.
Fix a $\delta \in (0, 1]$. 
Suppose the regularization parameter is such that
$\l \leq \tilde{\mathcal{O}}\left(\frac{n_x(n_y + n_u) + \ln(1/\delta)}{NL}\right)$.
For any global optimal points $(n_x, \{\hat{a}_j\}, \hat{B}, \hat{C})$ of 
program \eqref{eq:P1} satisfies Assumption \ref{ass:local_lip}. 
If $N/\ln(NL) \gtrsim {\sf cond}_{\F_{\theta}^{\eqref{eq:P1}}}^2\times[n_x(n_y + n_u) + \ln(1/\delta)]$,
then w.p at-least $1-\delta$ we have that
\be
\nmm{(G(\diag(\{\hat{a}_j\}), \hat{B}, \hat{C}) - G^*){\tilde{\Sigma}_U}^{1/2}}^2_{F}\hspace{-2pt}
\leq \tilde{\mathcal{O}}\left(\left(\sigma^2  + \Omega^{\eqref{eq:P1}}(G^*)\right)\left[\frac{n_x(n_y + n_u) + \ln(1/\delta)}{NL}\right]\right).
\ee
\end{theoremA}

\begin{proof}
The proof technique is very similar to \textsection\ref{sec:apdx_thm_p2_stat}. With an exception that
the dimension of the parameters ${\sf dim}(\F_{\theta}^{\eqref{eq:P1}})$ with assumption 
\ref{ass:local_lip}
is
$n_x(n_y+n_x + 1)$. 
\end{proof}

\subsection{Proof of Corollary \ref{crl:crl_p1_stat}}\label{sec:apdx_crl_p1_stat}
\begin{corollaryA}
Under the conditions of Theorem \ref{thm:thm_p1_stat}, the following statements holds true,
\begin{enumerate}
    \item If $A^*$ is not real, diagonalizable then the upper bound of Equation \ref{eq:stat_p1} evaluates to $\infty$.
    \item Otherwise, the sample complexity for the recovery of $G^*$ with program \eqref{eq:P1} is
    $NL \gtrsim [n_x(n_y + n_u) + \ln(1/\delta)] \cdot {\sf polylog}(NL)$.
\end{enumerate}
\end{corollaryA}
\begin{proof}

\textbf{Statement (i):} Recall the definition of the optimal regularizer,
\be
\Omega^{\eqref{eq:P1}}(\hat{G}) := \inf_{n_x, \a, B, C \in \F_{\theta}^{\eqref{eq:P1}}: G(\diag(\a), B, C) = \hat{G}}\Theta_{n_x}^{\eqref{eq:P1}}(\a, B, C, D),
\ee
Suppose $A^*$ was not real, diagonizable then there does not exist a 
$(n_x, \a, B, C) \in \F_{\theta}^{\eqref{eq:P1}}$ that can make $G(\diag(\a), B, C) = G^*$. This make the infimum infeasible
making $\Omega^{\eqref{eq:P1}}(G^*) = \infty$. Therefore, the right
side term of Equation \ref{eq:stat_p1} is infinity.

\textbf{Statement (ii):} From the Equation \ref{eq:stat_p1} we can observe that it is sufficient,
\be
NL \gtrsim [n_x(n_y + n_u) + \ln(1/\delta)] \cdot {\sf polylog}(NL)
\ee
to ensure the right side term decays asymptotically.
\end{proof}

\subsection{Proof of Theorem \ref{thm:apdx_timedep_risk}}\label{sec:apdx_timedep_risk_proof}

Before diving into the proof of Theorem \ref{thm:apdx_timedep_risk} let us 
state few key lemmas the will help construct our result.
First we link empirical recovery error with the noise term
and regularizer via Lemma \ref{lemma:opt_lemma}.
\begin{lemma}\label{lemma:opt_lemma}
Suppose $\hat{\theta} \in \F_{\theta}$ is the global minimizer of the Equation \eqref{eq:emp_obj}, then we have
\be\label{eq:basic_ineq}
\frac{1}{NT}\sum_{i=1}^{N}\sum_{t=1}^{T}\nmm{f_{\hat{\theta}}(\x_t^i) - f^*(\x_t^i)}_{2}^2  \leq
\frac{2}{NT}\sum_{i=1}^{N}\sum_{t=1}^{T}\IP{\v_t^i}{f_{\hat{\theta}}(\x_t^i) - f^*(\x_t^i)} + \l \left[\Omega(f^*) - \Omega(f_{\hat{\theta}})\right].
\ee
\end{lemma}

\begin{proof}
By the definition of the global minimizer of ERM, we have that,
\be
\frac{1}{NT}\sum_{i=1}^{N}\sum_{t=1}^{T}\nmm{\y_t^i - f_{\hat{\theta}}(\x_t^i)}_2^2 + \l \Omega(f_{\hat{\theta}})
\leq \frac{1}{NT}\sum_{i=1}^{N}\sum_{t=1}^{T}\nmm{\y_t^i - f^*(\x_t^i)}_2^2 + \l \Omega(f^*).
\ee
From data the data generating mechanism we have that,
\be
\begin{split}
\frac{1}{NT}&\sum_{i=1}^{N}\sum_{t=1}^{T}\nmm{f^*(\x_t^i) + \v_t^i - f_{\hat{\theta}}(\x_t^i)}_2^2 + \l \Omega(f_{\hat{\theta}})\\
&\leq \frac{1}{NT}\sum_{i=1}^{N}\sum_{t=1}^{T}\nmm{f^*(\x_t^i) + \v_t^i - f^*(\x_t^i)}_2^2 + \l \Omega(f^*),
\end{split}
\ee
upon simplifying it further we have,
\be
\begin{split}
\frac{1}{NT}\sum_{i=1}^{N}\sum_{t=1}^{T}\nmm{f^*(\x_t^i) - f_{\hat{\theta}}(\x_t^i)}_2^2 &+ \nmm{\v_t^i}_2^2 + 2\IP{\v_t^i}{f^*(\x_t^i) - f_{\hat{\theta}}(\x_t^i)}\\
&\leq \frac{1}{NT}\sum_{i=1}^{N}\sum_{t=1}^{T}\nmm{\v_t^i}_2^2 + \l \left[\Omega(f^*) - \Omega(f_{\hat{\theta}})\right],
\end{split}
\ee
which evaluates to,
\be
\begin{split}
\frac{1}{NT}\sum_{i=1}^{N}\sum_{t=1}^{T}&\nmm{f^*(\x_t^i) - f_{\hat{\theta}}(\x_t^i)}_2^2 \\
&\leq \frac{2}{NT}\sum_{i=1}^{N}\sum_{t=1}^{T}\IP{\v_t^i}{f_{\hat{\theta}}(\x_t^i)-f^*(\x_t^i)}
+ \l \left[\Omega(f^*) - \Omega(f_{\hat{\theta}})\right].
\end{split}
\ee
This concludes the proof.
\end{proof}
Next we link empirical recovery loss with population loss
for a fixed hypothesis class. We will directly apply
the result from \cite{ziemann_tutorial_2024}.
\begin{lemma}\label{lemma:cond_fixed}\cite[Lemma F.2]{ziemann_tutorial_2024}
Let $\x_{t}^{i}$ be $N$ iid samples. Define
\be
{\mathsf cond}_{\F_{\theta}} := \max_{\theta \in \F_{\theta}}
\max_{t \in [T]}\frac{\sqrt{\mathbb{E}\left[\nmm{f_{\theta}(\x_t)}_2^4\right]}}{\mathbb{E}\left[\nmm{f_{\theta}(\x_t)}_2^2\right]}.
\ee
For any $\theta \in \F_{\theta}$ and $u \geq 0$ we have 
\be
\mathbb{P}\left(\sum_{i=1}^{N}\sum_{t=1}^{T}\nmm{f_{\theta}(\x_t^i)}_2^2
< \frac{1}{2}\sum_{i=1}^{N}\sum_{t=1}^{T}\mathbb{E}\left[\nmm{f_{\theta}(\x_t^i)}_2^2\right]\right)
\leq \exp\left(\frac{-N}{4{\mathsf{cond}}_{\F_{\theta}}^2}\right).
\ee
\end{lemma}
Lemma~\ref{lemma:cond_fixed} applies to a fixed function. 
Although~\cite{ziemann_tutorial_2024} provides a uniform bound, it 
is restricted to finite hypothesis classes. The work 
of~\cite{ziemann-et-al-icml24} extends this to infinite classes but 
focuses on Orlicz spaces and non-realizable settings. In contrast, 
our analysis is simpler; we establish a uniform bound for sub-Gaussian random variables, 
directly linking empirical and population 
excess risk.
\begin{lemma}\label{lemma:cond}
Let $\x_{t}^{i}$ be $N$ iid samples such
that for a convex set $\C_x \subseteq {\R^{n_x}}^{\otimes T}$ it holds 
true that $\mathbb{P}\left(\cap_{i\in[N]}\x_t^i\right) \leq \delta_{\C_x}$. 
Define
\be
{\mathsf cond}_{\F_{\theta}} := \max_{\theta \in \F_{\theta}}
\max_{t \in [T]}\frac{\sqrt{\mathbb{E}\left[\nmm{f_{\theta}(\x_t)}_2^4\right]}}{\mathbb{E}\left[\nmm{f_{\theta}(\x_t)}_2^2\right]},
\ee
and for a fixed $f^*: \R^{n_x} \to \R^{n_y}$
\be
D_2(\C_x) := \sup_{\theta \in \F_{\theta}}
\left|\nmm{f_{\theta} - f^*}_{L_2}^2 -
\nmm{f_{\theta}\circ \P_{\C_x} - f^* \circ \P_{\C_x}}_{L_2}^2\right|.
\ee
If for all $\x \in \C_{x}$ we have
that for some metric $d$ in $\F_{\theta}$
\be
\forall \theta, \theta' \in \F_{\theta}:
\nmm{f_{\theta}(\x) - f_{\theta'}(\x)} \leq 
K_{\C_x}d(\theta, \theta').
\ee
and $\nmm{f_{\theta}(\z)-f^*(\z)}_2 \leq B_{\C_x}$.
Then for any $u \geq 0$ we have that
\be
\begin{split}
&\mathbb{P}\left(\frac{1}{NT}\sum_{i=1}^{N}\sum_{t=1}^{T}\nmm{f_{\theta}(\x_t)}_2^2
- \frac{1}{2T}\sum_{t=1}^{T}\mathbb{E}\left[\nmm{f_{\theta}(\x_t)}_2^2\right] \leq u + D_2(\C)\right)\\
&\leq \exp\left(\ln(\N(\F_{\theta}, d, u/4K_{\C_x}B_{\C_x}))-\frac{N}{4{\sf cond}_{\F_{\theta}}^2}\right).
\end{split}
\ee
\end{lemma}

\begin{proof}
We proceed with convex projection analysis similar to 
\cite{tadipatri-et-al-iclr25}. Denote 
$Z_i = \begin{bmatrix}\z_1^i & \cdots & \z_T^i \end{bmatrix}$ and define
\be
g_{\theta}([Z_i]_{i \in [N]})
:= 
-\frac{1}{NT}\sum_{i=1}^{N}\sum_{t=1}^{T}\nmm{f_{\theta}(\z_t)-f^*(\z_t)}_2^2
+\frac{1}{2T}\sum_{t=1}^{T}\mathbb{E}\left[\nmm{f_{\theta}(\z_t)-f^*(\z_t)}_2^2\right].
\ee
Now we estimate the Lipschtiz constant $g_{\theta}$ 
in $\F_{\theta}$.
\be
\begin{split}
g_{\theta}([Z_i]_{i \in [N]}) - g_{\theta'}([Z_i]_{i \in [N]}) 
&= -\frac{1}{NT}\sum_{i=1}^{N}\sum_{t=1}^{T}\nmm{f_{\theta}(\z_t)-f^*(\z_t)}_2^2-\nmm{f_{\theta'}(\z_t)-f^*(\z_t)}_2^2\\
&+\frac{1}{2T}\sum_{t=1}^{T}\mathbb{E}\left[\nmm{f_{\theta}(\z_t)-f^*(\z_t)}_2^2-\nmm{f_{\theta'}(\z_t)-f^*(\z_t)}_2^2\right].
\end{split}
\ee
From the problem setting we obtain that
\be
\begin{split}
|g_{\theta}([Z_i]_{i \in [N]}) - g_{\theta'}([Z_i]_{i \in [N]})| 
&\leq 4K_{\C}B_{\C}d(\theta, \theta').
\end{split}
\ee
In the similar analysis as Lemma \ref{lemma:sup_hypo} we
obtain for any $u \geq 0$
\be
\begin{split}
\mathbb{P}\left(\sup_{\theta \in \F_{\theta}}g_{\theta}([X_i]_{i\in[N]}) > u + D_2(\C) \right)
\leq \N(\F_{\theta}, d, u/2K'(\C)) \cdot \mathbb{P}\left(g_{\theta'}(\P_{\C}([X_i]_{i\in[N]})) > 0 \right).
\end{split}
\ee
From Lemma \ref{lemma:cond} we obtain
\be
\begin{split}
\mathbb{P}\left(\sup_{\theta \in \F_{\theta}}g_{\theta}([X_i]_{i\in[N]}) > u + D_2(\C) \right)
\leq \N(\F_{\theta}, d, u/4K_{\C_x}B_{\C_x}) \exp\left(-\frac{N}{4{\sf cond}_{\F_{\theta}}^2}\right).
\end{split}
\ee
\end{proof}
Up to this point we have linked empirical and population
excess risk uniformly. Now we prove the uniform concentration
of residual error in Equation \eqref{eq:basic_ineq}.
\begin{lemma}\label{lemma:sup_hypo}
Suppose $\v_t|\x_{1:t}$ is sub-gaussian vector with proxy variance $\sigma^2I$
in $\R^{n_y}$, where $\x_t \in \R^{n_x}$. 
Consider the parameter class $\F_{\theta}$,
and corresponding function $f_{\theta}: \R^{n_x} \to \R^{n_y}$
for all $\theta \in \F_{\theta}$. Let $f^*$ be some fixed realizable function. 
Suppose we have access to $N$ iid samples $(\x_{t}^i, \v_{t}^i)$ such that there is a convex set
$\C \subseteq {\R^{n_x}}^{\otimes T} \times {\R^{n_y}}^{\otimes T}$ which satisfies $P(\cap_{i\in[N]}\cap_{t\in [T]}(\x_t^i, \v_t^i) \notin \C) \leq \delta_{\C}$.
Define a real-valued function
\be
M(f, [(\x_t^i, \v_t^i)_{t\in[T]}]_{i \in [N]}) := \frac{1}{NT}\sum_{i=1}^{N}\sum_{t=1}^{T}\IP{f(\x_t^i)}{\v_t^i-f(\x_t^i)},
\ee
projection gap
\be
D(\C) := \sup_{\theta \in \F_{\theta}}\left|\mathbb{E}\left[M(f_{\theta}-f^*, [(\x_t^i, \v_t^i)_{t\in[T]}]_{i \in [N]})-M(f_{\theta}-f^*, \P_{\C}([(\x_t^i, \v_t^i)_{t\in[T]}]_{i \in [N]}))\right]\right|.
\ee
If for all $((\x_t, \v_t)_{t\in [T]}) \in \C$ for some metric $d$ defined in $\F_{\theta}$ it holds true that
\be
\forall \theta, \theta' \in \F_{\theta}; \nmm{f_{\theta}(\x_t) - f_{\theta'}(\x_t)} \leq K_{\C}d(\theta, \theta'),
\ee
and
\be
\max\{\nmm{f_{\theta}(\x_t) - f^*(\x_t)}_2, \nmm{\v_t}_2\} \leq B_{\C}.
\ee
Then we have that
\be
\begin{split}
&\mathbb{P}\left(\sup_{\theta \in \F_{\theta}}M(f_{\theta}-f^*[(\x_t^i, \v_t^i)_{t\in[T]}]_{i \in [N]}) > u + D(\C)\right)\\
&\leq \delta_{\C} + \exp\left(\ln\left(\N\left(\F_{\theta}, d, \frac{u}{6K_{\C}B_{\C}}\right)\right) - \frac{NT}{\sigma^2}u\right).
\end{split}
\ee
\end{lemma}

\begin{proof}
The proof is similar to Proposition F.2 in \cite{ziemann_tutorial_2024}. However, we follow the style of 
\cite{tadipatri-et-al-iclr25} proof analysis via projecting onto convex sets.

Denote $Z_i = \begin{bmatrix} \z_1^i & \cdots & \z_T^i \end{bmatrix}$, and $W_i = \begin{bmatrix} \w_1^i & \cdots & \w_T^i \end{bmatrix}$. Define, 
\be
g_{\theta}([Z_i, W_i]_{i\in[N]}) = \frac{1}{NT}\sum_{t=1}^{T}\IP{f_{\theta}(\z_t^i) - f^*(\z_t^i)}{\w_t^i - \left[f_{\theta}(\z_t^i) - f^*(\z_t^i)\right]}.
\ee
We begin to estimate the local Lipschitz constant of $g_{\theta}$ in $\C$. First
choose any $[Z, W] \in \C$ since we have,
\be
\forall \theta, \theta' \in \F_{\theta}; \nmm{f_{\theta}(\z_t) - f_{\theta'}(\z_t)} \leq K_{\C}d(\theta, \theta').
\ee
Then for any $[Z_i, W_i] \in \C$ we have,
\begin{eqnarray}
\left|{g_{\theta}([Z_i, W_i]_{i\in[N]}) - g_{\theta'}([Z_i, W_i]_{i\in[N]})}\right| &=&\hspace{-8pt}
\big|\frac{1}{NT}\sum_{i=1}^{N}\sum_{t=1}^{T}\IP{f_{\theta}(\z_t^i) - f^*(\z_t^i)}{\w_t^i - \left[f_{\theta}(\z_t^i) - f^*(\z_t^i)\right]} \nonumber \\
&-&\hspace{-8pt} \IP{f_{\theta'}(\z_t^i) - f^*(\z_t^i)}{\w_t^i - \left[f_{\theta'}(\z_t^i) - f^*(\z_t^i)\right]} \big|,\\
&=&\hspace{-8pt} \big|\frac{1}{NT}\sum_{i=1}^{N}\sum_{t=1}^{T}\IP{f_{\theta}(\z_t^{i}) - f_{\theta'}(\z_t^{i})}{\w_t^i - \left[f_{\theta}(\z_t^{i}) - f^*(\z_t^{i})\right]} \nonumber \\
&+&\hspace{-8pt} \IP{f_{\theta'}(\z_t^{i}) - f^*(\z_t^i)}{f_{\theta'}(\z_t^i) - f_{\theta}(\z_t^i)} \big|,
\end{eqnarray}
First we apply triangular inequality and then Cauchy-Schwartz inequality then obtain,
\be
\begin{split}
\left|{g_{\theta}([Z_i, W_i]_{i\in[N]}) - g_{\theta'}([Z_i, W_i]_{i\in[N]})}\right| \leq
\frac{1}{NT}\sum_{i=1}^{N}\sum_{t=1}^{T}&\nmm{f_{\theta}(\z_t^{i}) - f_{\theta'}(\z_t^{i})} \times \Big[\nmm{f_{\theta'}(\z_t^{i}) - f^*(\z_t^{i})}\\
& + \nmm{\w_t^{i} - \left[f_{\theta}(\z_t^{i}) - f^*(\z_t^{i})\right]}\Big].
\end{split}
\ee
Since $\forall [Z, W] \in \C$, and $\theta \in \F_{\theta}$ it is the case that
$\max\{\nmm{f_{\theta}(\z_t) - f^*(\z_t)}, \nmm{\w_t}\} \leq B_{\C}$, via this assumption we obtain,
\be
\left|{g_{\theta}([Z_i, W_i]_{i\in[N]}) - g_{\theta'}([Z_i, W_i]_{i\in[N]})}\right| \leq \underbrace{3K_{\C}B_{\C}}_{=:K'(\C)}d(\theta, \theta').
\ee
Suppose that, $\theta^*$ is the maximizer of $g_{\theta}([Z_i, W_i]_{i\in[N]})$ then $\exists \theta_0 \in \N_{\epsilon}(\F_{\theta}, d)$
such that $d(\theta^*, \theta_0) \leq \epsilon$ therefore we have,
\be
-K'(\C)\epsilon \leq \sup_{\theta \in \F_{\theta}}g_{\theta}([Z_i, W_i]_{i\in[N]}) - g_{\theta_0}([Z_i, W_i]_{i\in[N]}) \leq K'(\C)\epsilon.
\ee
By the monotonicity of the probability measure we have,
\be
\mathbb{P}\left(g_{\theta_0}([Z_i, W_i]_{i\in[N]}) + K'(\C)\epsilon > u \right) \geq 
\mathbb{P}\left(\sup_{\theta \in \F_{\theta}}g_{\theta}([Z_i, W_i]_{i\in[N]}) > u \right),
\ee
further more by applying supremum over all the $\epsilon$-net then we have,
\be
\mathbb{P}\left(\sup_{\theta_0 \in \N_{\epsilon}(\F_{\theta}, d)}g_{\theta_0}([Z_i, W_i]_{i\in[N]}) > u - K'(\C)\epsilon \right) \geq 
\mathbb{P}\left(\sup_{\theta \in \F_{\theta}}g_{\theta}([Z_i, W_i]_{i\in[N]}) > u \right).
\ee
Now set $\epsilon = u/2K'(\C)$ then we have,
\be
\mathbb{P}\left(\sup_{\theta_0 \in \N_{u/2K'(\C)}(\F_{\theta}, d)}g_{\theta_0}([Z_i, W_i]_{i\in[N]}) > u/2 \right) \geq 
\mathbb{P}\left(\sup_{\theta \in \F_{\theta}}g_{\theta}([Z_i, W_i]_{i\in[N]}) > u \right),
\ee
by sub-additivity of probability measure we have,
\be
\sum_{\theta_0 \in \N_{u/2K'(\C)}(\F_{\theta}, d)}\mathbb{P}\left(g_{\theta_0}([Z_i, W_i]_{i\in[N]}) > u/2 \right) \geq 
\mathbb{P}\left(\sup_{\theta \in \F_{\theta}}g_{\theta}([Z_i, W_i]_{i\in[N]}) > u \right).
\ee
Now choose a fixed $\theta' \in \N_{u/2K'(\C)}(\F_{\theta}, d)$ we have
\be\label{eq:c1_1}
\mathbb{P}\left(\sup_{\theta \in \F_{\theta}}g_{\theta}([Z_i, W_i]_{i\in[N]}) > u \right)
\leq \N(\F_{\theta}, d, u/2K'(\C)) \cdot \mathbb{P}\left(g_{\theta'}([Z_i, W_i]_{i\in[N]}) > u/2 \right).
\ee

Observe that for $[X_i, V_i]_{i\in[N]}$ (not necessarily in the convex set $\C$) we have that

\be
\begin{split}
\mathbb{P}\big(&g_{\theta}([X_i, V_i]_{i\in[N]}) + \mathbb{E}\left[g_{\theta}([X_i, V_i]_{i\in[N]})\right]
= g_{\theta}(\P_{\C}([X_i, V_i]_{i\in[N]})) + \mathbb{E}\left[g_{\theta}(\P_{\C}([X_i, V_i]_{i\in[N]}))\right]\big) \\
&\geq 1 - \delta_{\C}.
\end{split}
\ee
By our assumptions we have that
\be
\mathbb{E}\left[g_{\theta}(\P_{\C}([X_i, V_i]_{i\in[N]}))\right] - 
\mathbb{E}\left[g_{\theta}([X_i, V_i]_{i\in[N]})\right] \leq B_{r}.
\ee
Therefore we have
\be\label{eq:c1_2}
\begin{split}
\mathbb{P}\big(g_{\theta}([X_i, V_i]_{i\in[N]}) \leq g_{\theta}(\P_{\C}([X_i, V_i]_{i\in[N]})) + B_{r}\big) 
\geq 1 - \delta_{\C}.
\end{split}
\ee

First we apply supremum over $\theta$, then set $[Z_i, V_i] = \P_{\C}([X_i, V_i])$. Finally by
combining Equations \eqref{eq:c1_1}, and \eqref{eq:c1_2} we obtain
\be
\begin{split}
\mathbb{P}\left(\sup_{\theta \in \F_{\theta}}g_{\theta}([X_i, V_i]_{i\in[N]}) > u + D(\C) \right)
\leq \N(\F_{\theta}, d, u/2K'(\C)) \cdot \mathbb{P}\left(g_{\theta'}(\P_{\C}([X_i, V_i]_{i\in[N]})) > u/2 \right).
\end{split}
\ee
By replacing all the variables back we have,
\be\label{eq:c1}
\begin{split}
&\mathbb{P}\Big(\sup_{\theta \in \F_{\theta}}\frac{1}{NT}\sum_{i=1}^{N}\sum_{t=1}^{T}\IP{f_{\theta}(\x_t^i)-f^*(\x_t^i)}{\v_t^i - \left[f_{\theta}(\x_t^i)-f^*(\x_t^i)\right]}
> u + D(\C)\Big) \leq \delta_{\C}\\
&+ \N(\F_{\theta}, d, u/2K'(\C)) \cdot
\mathbb{P}\left(\frac{1}{NT}\sum_{i=1}^{N}\sum_{t=1}^{T}\IP{f_{\theta'}(\x_t^i)-f^*(\x_t^i)}{\v_t - \left[f_{\theta'}(\x_t^i)-f^*(\x_t^i)\right]} > u/2\right).
\end{split}
\ee
At this stage we have linked supremum process with a fixed parameter in $\epsilon$-net of $\F_{\theta}$.
Next we compute the probability of the event in the right side term of Equation \eqref{eq:c1} for a fixed $\theta'$.
By applying Cram{\'e}r-Chernoff bound for any $a \geq 0$,
\be
\begin{split}
&\mathbb{P}\left(\frac{1}{NT}\sum_{i=1}^{N}\sum_{t=1}^{T}\IP{f_{\theta'}(\x_t^i)-f^*(\x_t^i)}{\v_t^i - \left[f_{\theta'}(\x_t^i)-f^*(\x_t^i)\right]} > u\right)\\
&\leq e^{-au}\mathbb{E}\left[\exp\left(a\frac{1}{NT}\sum_{i=1}^{N}\sum_{t=1}^{T}\IP{f_{\theta'}(\x_t^i)-f^*(\x_t^i)}{\v_t^i - \left[f_{\theta'}(\x_t^i)-f^*(\x_t^i)\right]}\right)\right].
\end{split}
\ee
By independence we have that
\be
\begin{split}
&\mathbb{P}\left(\frac{1}{NT}\sum_{i=1}^{N}\sum_{t=1}^{T}\IP{f_{\theta'}(\x_t^i)-f^*(\x_t^i)}{\v_t^i - \left[f_{\theta'}(\x_t^i)-f^*(\x_t^i)\right]} > u\right)\\
&\leq e^{-au}\prod_{i=1}^{N}\mathbb{E}\left[\exp\left(a\frac{1}{NT}\sum_{t=1}^{T}\IP{f_{\theta'}(\x_t^i)-f^*(\x_t^i)}{\v_t^i - \left[f_{\theta'}(\x_t^i)-f^*(\x_t^i)\right]}\right)\right].
\end{split}
\ee

Recall the martingale property of conditional
independence for a random variables $a_{1:T}$
we have that $\mathbb{E}\left[\left(\prod_{t=1}^{T}a_t\right)\right] = \mathbb{E}\left[\left(\prod_{t=1}^{T-1}a_t\right) \times \mathbb{E}\left[a_T \mid a_{1:(T-1)}\right]\right]$. In the similar spirit we obtain
\be
\begin{split}
&\mathbb{P}\left(\frac{1}{NT}\sum_{i=1}^{N}\sum_{t=1}^{T}\IP{f_{\theta'}(\x_t^i)-f^*(\x_t^i)}{\v_t^i - \left[f_{\theta'}(\x_t^i)-f^*(\x_t^i)\right]} > u\right)\\
&\leq e^{-au}\prod_{i=1}^{N}\mathbb{E}\Big[\exp\left(a\frac{1}{NT}\sum_{t=1}^{T-1}\IP{f_{\theta'}(\x_t^i)-f^*(\x_t^i)}{\v_t^i - \left[f_{\theta'}(\x_t^i)-f^*(\x_t^i)\right]}\right)\\ 
& \times \mathbb{E}\left[\exp\left(\frac{a}{NT}\IP{f_{\theta'}(\x_T^i)-f^*(\x_T^i)}{\v_T^i} - \frac{a}{NT}\nmm{\left[f_{\theta'}(\x_T^i)-f^*(\x_T^i)\right]}_2^2\right) \mid (\x_{1:T-1}^i, \v_{1:T-1}^i)\right] \Big].\\
\end{split}
\ee
From our assumption we have that $\v_T^i | \x_{1:T}^i$ is a sub-Gaussian vector then for any
vector $\y$, and $\l \geq 0$ we have that $\mathbb{E}[\exp(\l \IP{\y}{\v_T^i}) \mid \x_{1:T}^i] \leq 
\exp(0.5\sigma^2\l^2\nmm{\y}_2^2)$. Then applying this property to earlier conditional expectation 
we obtain,
\be
\begin{split}
&\mathbb{P}\left(\frac{1}{NT}\sum_{i=1}^{N}\sum_{t=1}^{T}\IP{f_{\theta'}(\x_t^i)-f^*(\x_t^i)}{\v_t^i - \left[f_{\theta'}(\x_t^i)-f^*(\x_t^i)\right]} > u\right)\\
&\leq e^{-au}\prod_{i=1}^{N}\mathbb{E}\Big[\exp\left(a\frac{1}{NT}\sum_{t=1}^{T-1}\IP{f_{\theta'}(\x_t^i)-f^*(\x_t^i)}{\v_t^i - \left[f_{\theta'}(\x_t^i)-f^*(\x_t^i)\right]}\right)\\ 
& \times \exp\left(\frac{a^2}{2N^2T^2}\sigma^2\nmm{f_{\theta'}(\x_T^i)-f^*(\x_T^i)}_2^2 - \frac{a}{NT}\nmm{\left[f_{\theta'}(\x_T^i)-f^*(\x_T^i)\right]}_2^2\right) \Big].
\end{split}
\ee
Now repeated apply these argument for all $t \in [T-1]$ then we have
\be
\begin{split}
&\mathbb{P}\left(\frac{1}{T}\sum_{t=1}^{T}\IP{f_{\theta'}(\x_t)-f^*(\x_t)}{\v_t - \left[f_{\theta'}(\x_t)-f^*(\x_t)\right]} > u\right)\\
&\leq e^{-au} \prod_{i=1}^{N}\exp\left(\frac{a^2}{2N^2T^2}\sigma^2\nmm{f_{\theta'}(\x_t)-f^*(\x_t)}_2^2 - \frac{a}{NT}\nmm{\left[f_{\theta'}(\x_t)-f^*(\x_t)\right]}_2^2\right) \\
&= e^{-au}\prod_{i=1}^{N}\exp\left(\sum_{t=1}^{T}\frac{a^2}{2N^2T^2}\sigma^2\nmm{f_{\theta'}(\x_t)-f^*(\x_t)}_2^2 - \frac{a}{NT}\nmm{\left[f_{\theta'}(\x_t)-f^*(\x_t)\right]}_2^2\right).
\end{split}
\ee
Recall that the choice of $a$ was arbitrary in $\R_+$ we set $a = \frac{2NT}{\sigma^2}$ then we have that
\be
\begin{split}
&\mathbb{P}\left(\frac{1}{NT}\sum_{i=1}^{N}\sum_{t=1}^{T}\IP{f_{\theta'}(\x_t^i)-f^*(\x_t^i)}{\v_t^i - \left[f_{\theta'}(\x_t^i)-f^*(\x_t^i)\right]} > u\right)\\
&\leq e^{-au}\prod_{i=1}^{N}\exp\left(\sum_{t=1}^{T}\frac{a^2}{2N^2T^2}\sigma^2\nmm{f_{\theta'}(\x_t^i)-f^*(\x_t^i)}_2^2 - \frac{a}{NT}\nmm{\left[f_{\theta'}(\x_t^i)-f^*(\x_t^i)\right]}_2^2\right)\\
&= \exp\left(-\frac{2NT}{\sigma^2}u\right).
\end{split}
\ee
Now by re-scaling $u$ to $u/2$ we obtain
\be\label{eq:c2}
\mathbb{P}\left(\frac{1}{NT}\sum_{i=1}^{N}\sum_{t=1}^{T}\IP{f_{\theta'}(\x_t^i)-f^*(\x_t^i)}{\v_t^i - \left[f_{\theta'}(\x_t^i)-f^*(\x_t^i)\right]} > u/2\right) \leq \exp\left(-\frac{NT}{\sigma^2}u\right).
\ee
Finally we combine Equation \eqref{eq:c1}, and \eqref{eq:c2} to obtain,
\be
\begin{split}
&\mathbb{P}\left(\sup_{\theta \in \F_{\theta}}\frac{1}{NT}\sum_{t=1}^{T}\IP{f_{\theta}(\x_t^i)-f^*(\x_t^i)}{\v_t^i - \left[f_{\theta}(\x_t^i)-f^*(\x_t^i)\right]} > u + D(\C)  \right)\\
&\leq \delta_{\C} + \exp\left(\ln\left(|\N(\F_{\theta}, d, u/2K'(\C))|\right) - \frac{uT}{\sigma^2}\right).
\end{split}
\ee
\end{proof}
Essentially we have finished (i) linking empirical and population
excess risk and (ii) concentration of residual error in Equation
\eqref{eq:basic_ineq}. Now we combine all the results from
Lemma \ref{lemma:cond}, and \ref{lemma:sup_hypo} 
below to obtain excess risk bounds for dependent data.
\begin{proof}[Proof of Theorem \ref{thm:apdx_timedep_risk}]
As $\hat{\theta}$ is the global minimizer of Equation \eqref{eq:emp_obj} from Lemma \ref{lemma:opt_lemma} we have,
\be
\frac{1}{T}\sum_{t=1}^{T}\nmm{f_{\hat{\theta}}(\x_t) - f(\x_t)}_{2}^2  \leq
\frac{2}{T}\sum_{t=1}^{T}\IP{\v_t}{f_{\hat{\theta}}(\x_t) - f(\x_t)} + \l \Omega(f^*).
\ee
Multiplying both sides by 2 and re-arranging $\frac{1}{T}\sum_{t=1}^{T}\nmm{f_{\hat{\theta}}(\x_t) - f(\x_t)}_{2}^2$
to right side of inequality we obtain,
\be
\frac{1}{T}\sum_{t=1}^{T}\nmm{f_{\hat{\theta}}(\x_t) - f(\x_t)}_{2}^2  \leq
\frac{4}{T}\sum_{t=1}^{T}\IP{\v_t}{f_{\hat{\theta}}(\x_t) - f^*(\x_t)} - \frac{1}{T}\sum_{t=1}^{T}\nmm{f_{\hat{\theta}}(\x_t) - f^*(\x_t)}_{2}^2 + 2\l \Omega(f^*).
\ee
As $f_{\hat{\theta}}$ is random variable dependent on the empirical samples we apply a supremum on right side of inequality
to make the bound tight
\be
\begin{split}
\frac{1}{T}\sum_{t=1}^{T}\nmm{f_{\hat{\theta}}(\x_t) - f(\x_t)}_{2}^2  \leq&
\sup_{\theta \in \F_{\theta}}\left[\frac{4}{T}\sum_{t=1}^{T}\IP{\v_t}{f_{{\theta}}(\x_t) - f^*(\x_t)} - \frac{1}{T}\sum_{t=1}^{T}\nmm{f_{{\theta}}(\x_t) - f^*(\x_t)}_{2}^2\right] \\
&+ 2\l \Omega(f^*).
\end{split}
\ee
Re-writing the earlier we have,
\be
\begin{split}
\frac{1}{T}\sum_{t=1}^{T}\nmm{f_{\hat{\theta}}(\x_t) - f(\x_t)}_{2}^2  \leq &
\sup_{\theta \in \F_{\theta}}\left[\frac{1}{T}\sum_{t=1}^{T}\IP{\left(f_{{\theta}}(\x_t) - f^*(\x_t)\right)}{4\v_t - \left(f_{{\theta}}(\x_t) - f^*(\x_t)\right)}\right] \\
&+ 2\l \Omega(f^*).
\end{split}
\ee
Applying Lemma \ref{lemma:sup_hypo} with Assumption \ref{ass:data}, and \ref{ass:hypo} we conclude that for any $u > 0$,
\be
\begin{split}
&\mathbb{P}\left(\frac{1}{T}\sum_{t=1}^{T}\nmm{f_{\hat{\theta}}(\x_t) - f(\x_t)}_{2}^2
\leq u + D(\C) + 2\l \Omega(f^*)\right) \\
&\geq 1 - \delta_{\C} -  \exp\left(\ln\left(\N(\F_{\theta}, d, u/6K_{\C}B_{\C})\right) - \frac{uNT}{16\sigma^2}\right).
\end{split}
\ee
Now we connect the earlier inequality to $\nmm{f_{\hat{\theta}} - f^*}_{L_2}^2$ via Lemma \ref{lemma:cond}, after
which we obtain,
\be
\begin{split}
&\mathbb{P}\left(\nmm{f_{\hat{\theta}} - f^*}_{L_2}^2
\leq 4u + 2D(\C) + D_1(\C_x) + 4\l \Omega(f^*)\right) \\
&\geq 1 - 2\delta_{\C} -  \exp\left(\ln\left(\N(\F_{\theta}, d, u/6K_{\C}B_{\C})\right)- \frac{uNT}{16\sigma^2}\right) \\
&-\exp\left(\ln\left(\N(\F_{\theta}, d, u/4K_{\C_x}B_{\C_x})\right)
- \frac{N}{4{\sf cond}_{\F_{\theta}}^2}\right).
\end{split}
\ee

We choose $\C, \C_x$ to be a Euclidean ball
of radius $r$, $\mathbb{B}(r)$.

Since $\v_t$ is conditionally sub-Gaussian, and from Assumptions
\ref{ass:data}, and \ref{ass:hypo} it is
easy to observe that $\delta_{\C} = \mathcal{O}(\exp(-r^2/\sigma^2))$ and $D(\mathbb{B}(r)), D_1(\mathbb{B}(r))$ $=$ 
$\mathcal{O}({\sf poly(r, n_x, n_y)}\exp(-r^2/\sigma^2))$. Furthermore,
for Lipschtiz parametric class we have that
$\N(\F_{\theta}, d, \varepsilon) = \mathcal{O}(1/\varepsilon^{{\sf dim}(\F_{\theta})})$.

We choose $r$ such that
\be
u \geq \max\{D(\mathbb{B}(r)), D_1(\mathbb{B}(r)\},
\ee
\be
\begin{split}
\delta_{\mathbb{B}_r} \leq 
\min \Big\{&\exp\left(\ln\left(\N(\F_{\theta}, d, u/6K_{\mathbb{B}_r}B_{\mathbb{B}_r})\right) - \frac{uNT}{16 \sigma^2}\right),\\
&\exp\left(\ln\left(\N(\F_{\theta}, d, u/4K_{\mathbb{B}_r}B_{\mathbb{B}_r})\right) - \frac{N}{4{\sf cond}_{\F_{\theta}}^2}\right)
\Big\}.
\end{split}
\ee
and 
\be
\sqrt{u} \geq 6K_{\mathbb{B}_r}D_{\mathbb{B}_r}.
\ee
Therefore we obtain
\be
\begin{split}
\mathbb{P}\left(\nmm{f_{\hat{\theta}} - f^*}_{L_2}^2
\leq 6u + 4\l \Omega(f^*)\right)
\geq &1 - 2\exp\left(\ln\left(\N(\F_{\theta}, d, \sqrt{u})\right) - \frac{uNT}{16\sigma^2}\right) \\
&-2\exp\left(\ln\left(\N(\F_{\theta}, d, \sqrt{u})\right) - \frac{N}{4{\sf cond}_{\F_{\theta}}^2}\right).
\end{split}
 \ee

Now set 
\be
u = \tilde{\mathcal{O}}\left(\sigma^2\frac{{\sf dim}(\F_{\theta}) + \ln(1/\delta)}{NT}\right),
\ee
when $N/\ln(NT) \gtrsim {\sf cond}_{\F_{\theta}}^2 \times {\sf dim}(\F_{\theta})$
then we obtain
\be
\mathbb{P}\left(\nmm{f_{\hat{\theta}} - f^*}_{L_2}^2
\leq \tilde{\mathcal{O}}\left(\sigma^2\frac{{\sf dim}(\F_{\theta}) + \ln(1/\delta)}{NT}\right) + \l \Omega(f^*)\right) \geq 1 - \delta.
\ee
\end{proof}

\section{Related works}\label{sec:apdx_related}
\label{sec:related_works}

The word ``\textit{System Identification}'' was first coined in \cite{zadeh-ieee62}.
SysID deals with determination of outputs of a/many system given the inputs based on previous observations
of inputs and outputs. In this section, we disscus only the key relevant works while 
the readers are referred to \cite{aastrom-automatica71, de-jcam00} for an extensive survey.
Estimation of impulse response of a linear system can be re-written as a least squares
problem whose solution is obtained Moore-Penrose pseudo inverse of inputs \cite{de-jcam00}.
However, with classical methods from impulse response estimation of system parameters 
required the knowledge of true-order of the system.

\textbf{Minimal-State Realization.} Due to the basis invariance of linear system parameters, and unknown
system order, it has been 
a great interest to identify structured realization of system parameters. In the seminal work of \cite{ho-kalman-66}
provided an algorithm to estimate lowest order realization of $(A, B, C)$ from the Markov parameters alone. 
This Ho-Kalman procedure requires to compute SVD of Hankel matrix with Markov parameters. This scales
quadratically with the trajectory length. N4SID method \citep{van-automatica94} performs
both least-squares to estimate the Markov parameters, and Ho-Kalman
procedure to estimate low-order system parameters. Similar SysID has been
viewed in frequency response by \cite{mcKelvey-tac96}.
As pointed out earlier, 
these methods does not scale for longer trajectories. There have been many works in literature under
name \textit{Subspace-Identification} (see \cite{qin-cce06} for detailed survey).

\textbf{Convex Relaxations.}  From the optimization lens of low-order
SysID is posed as Hankel matrix-rank minimization problem \citep{fazel-et-al-acc01},
while this NP-HARD to solve in general. Influential work of \cite{recht-et-al-cdc08}
proposed convex relaxations of matrix-rank minimization via
a nuclear-norm based regularization procedure that recovery guarantees under
certain conditions. This procedure was then adapted many linear SysID problems
in \cite{mohan-fazel-acc10, shah-et-al-cdc12, 
fazel-et-al-siam13, liu-et-al-scl13, chiuso-et-al-cdc13, pillonetto-et-automatica16, 
sun-et-al-ojcs22} and for non-linear SysID in \cite{markovsky-cdc17}. 
Equivalent convex relaxation were also performed in 
\cite{smith-tac14, honarpisheh-et-al-ifac24}. While these convex relaxation circumvents
the hardness of the problem, however, still necessitates the SVD computation
for the Hankel nuclear-norms.

\textbf{Nonconvex Relaxations.} To circumvent the SVD computation authors
of \cite{campi-et-al-cdc15} re-forumulated the Hankel-rank minimization
to estimating the amplitude, frequency, and phase of the system response
making the problem nonconvex and computationally faster. However, the adopted algorithm
is greedy and have only probabilistic global optimality guarantees. For faster learning of
system parameters, Deep Learning (DL) methods very adopted in \cite{pillonetto-et-al-automatica25}.
However, the main critique of optimality certificates for DL make it not realiable control
applications. But in the literature of Machine Learning there are nonconvex problems whose
global optimality certificates are well established \cite{haeffele_global_2017}. For instance,
the necessity of SVD computation 
for nuclear norm computation was greatly reduced by re-parameterzing the matrix
with a product of left and right matrices commonly known as
BM factorization \citep{burer-monteiro-mp03}. By doing this, the
nuclear-norm can written as variational minimization of sums of square Frobenous norms of 
left and right matrices \citep{recht-et-al-siam10, giampouras-et-al-nips20}. This trick
has been adopted to many structured matrix factorization problems in \cite{bach-arxiv13, haeffele-vidal-tpami20}.
These works form a basis for our proposed reformulations for SysID.

\textbf{Statistical Analysis.} Classical works
used maximum likelihood estimation to perform SysID
that provided asymptotic convergences 
\citep{lennart1980asymptotic}. However, they required
the prior information about the covariances
of the system parameters that posits a limitation
for unknown SysID. \cite{simchowitz2018learning}
analyzed recovery error rates but their
analysis was limited to fully observed system.
\cite{tsiamis2019finite}, and \cite{oymak2019non}
were the one of the first ones to provide error rates 
for partially observed systems. While
all the earlier mentioned works analyzed SysID without
the consideration of low-orderness. \cite{sun-et-al-ojcs22}
studied both optimization and statistical aspects
of low-order SysID, however the estimated sample
complexities scaled badly with the trajectory length. 
Later work on SVD truncation of Hankel matrices yielded a 
tighter error rates and sample complexities
\citep{lee2022improved}. We compare our estimates statistical error rates and sample complexities
with earlier mentioned works in Table \ref{tab:comparision_stat}.

\section{Extra numerical Simulations}\label{sec:apdx_numerical}
This section describes algorithmic details and more experimental results for the \textsection\ref{sec:numerical_simulations}.
\input{algorithms_details}
\begin{figure}[H]
    \centering
    \subfigure[Varying sample size]{
        \includegraphics[width=0.8\textwidth]{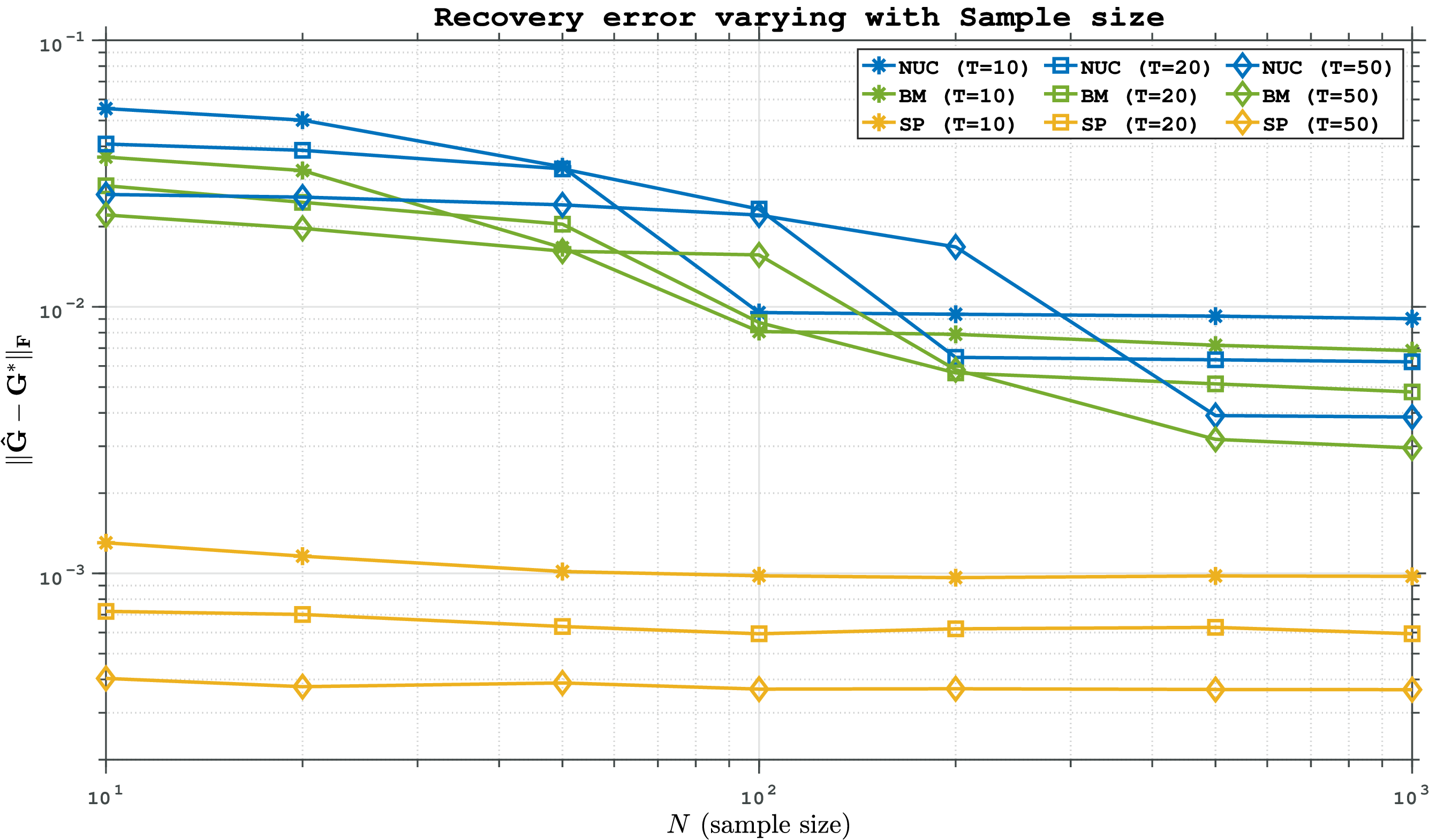}
        \label{fig:vfig1}
    }
    \subfigure[Varying trajectory length]{
        \includegraphics[width=0.8\textwidth]{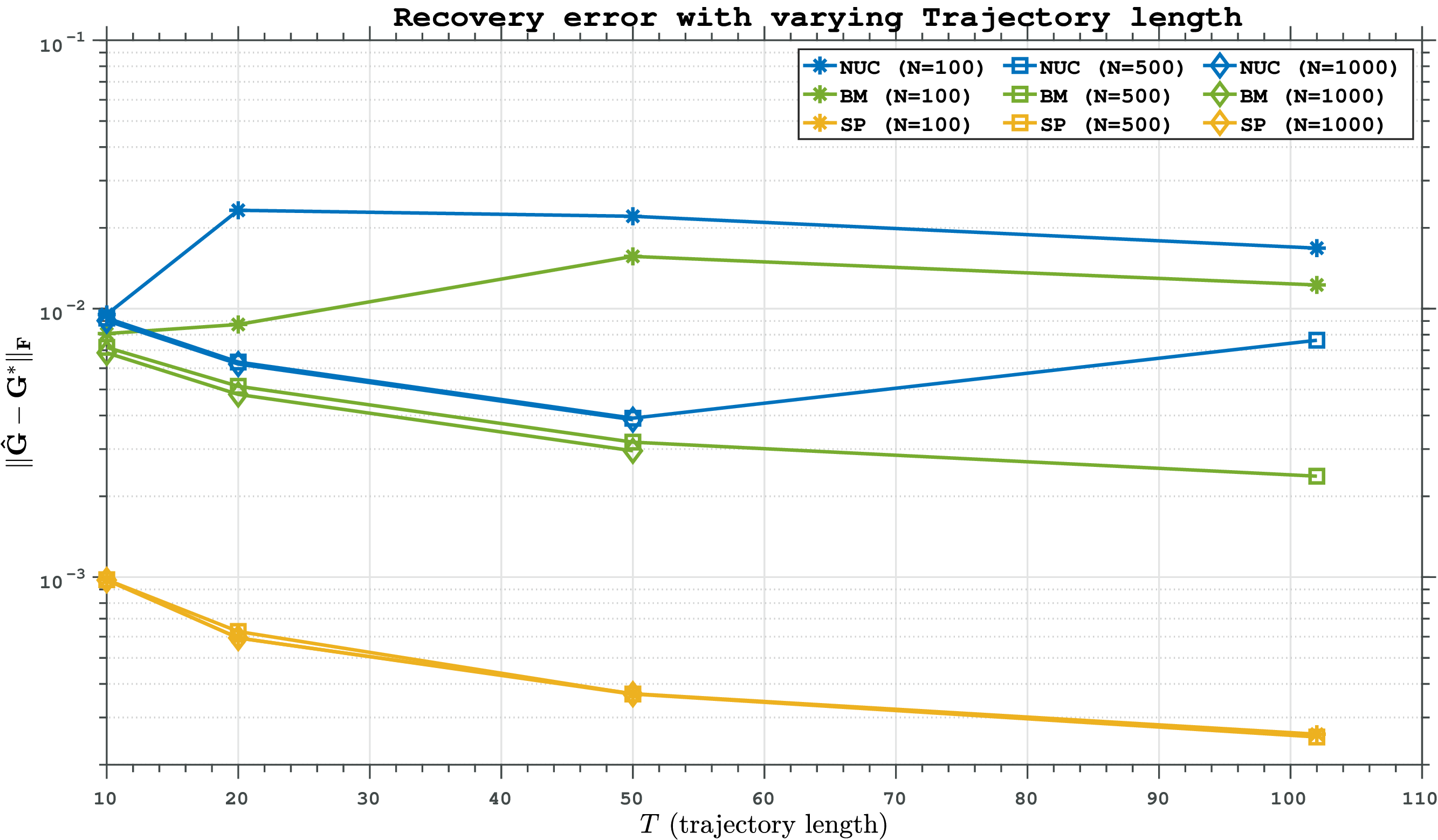}
        \label{fig:vfig2}
    }
    \caption{Varying sample size and trajectory length. \textsf{NUC} refers to \eqref{eq:P1_1},
    \textsf{BM} refers to \eqref{eq:P2}, and \textsf{SP} refers to \eqref{eq:P1}.}
    \label{fig:vary_sample_trajectory_length}
\end{figure}

\begin{figure}[H]
    \centering
    \subfigure[Recovery Error vs Time]{
        \includegraphics[width=0.48\textwidth]{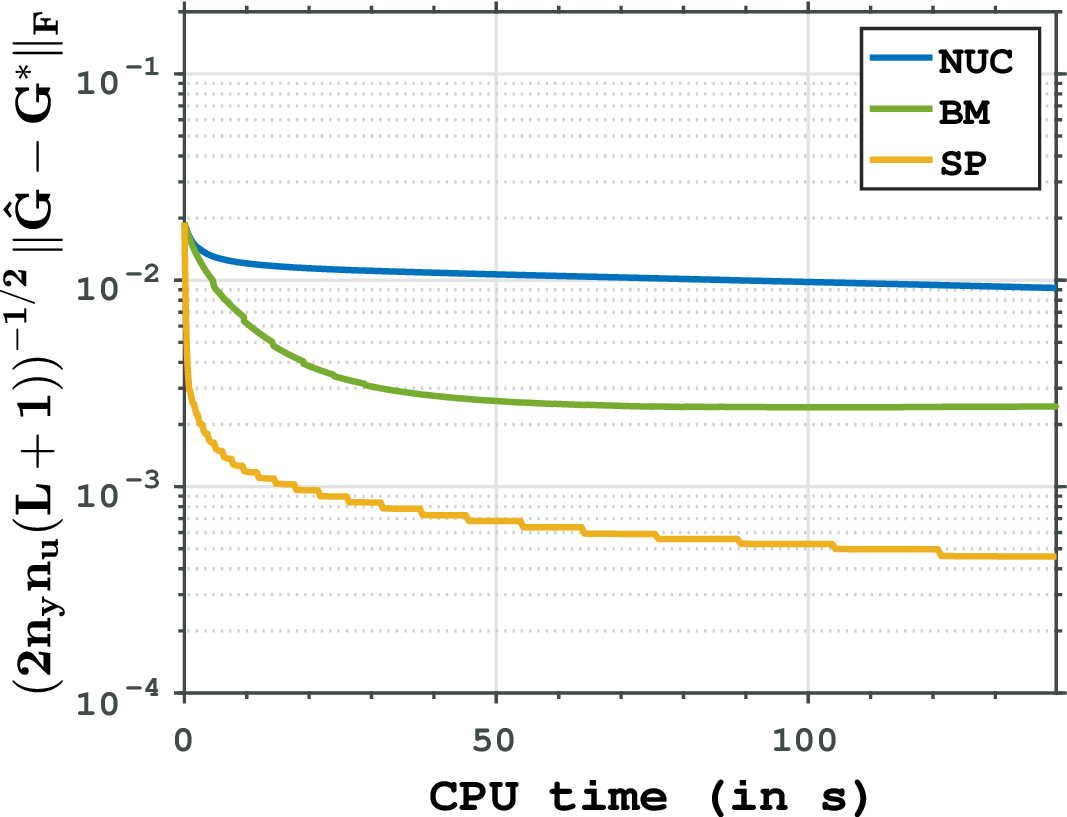}
        \label{fig:nfig1}
    }
    \subfigure[Loss vs Time]{
        \includegraphics[width=0.48\textwidth]{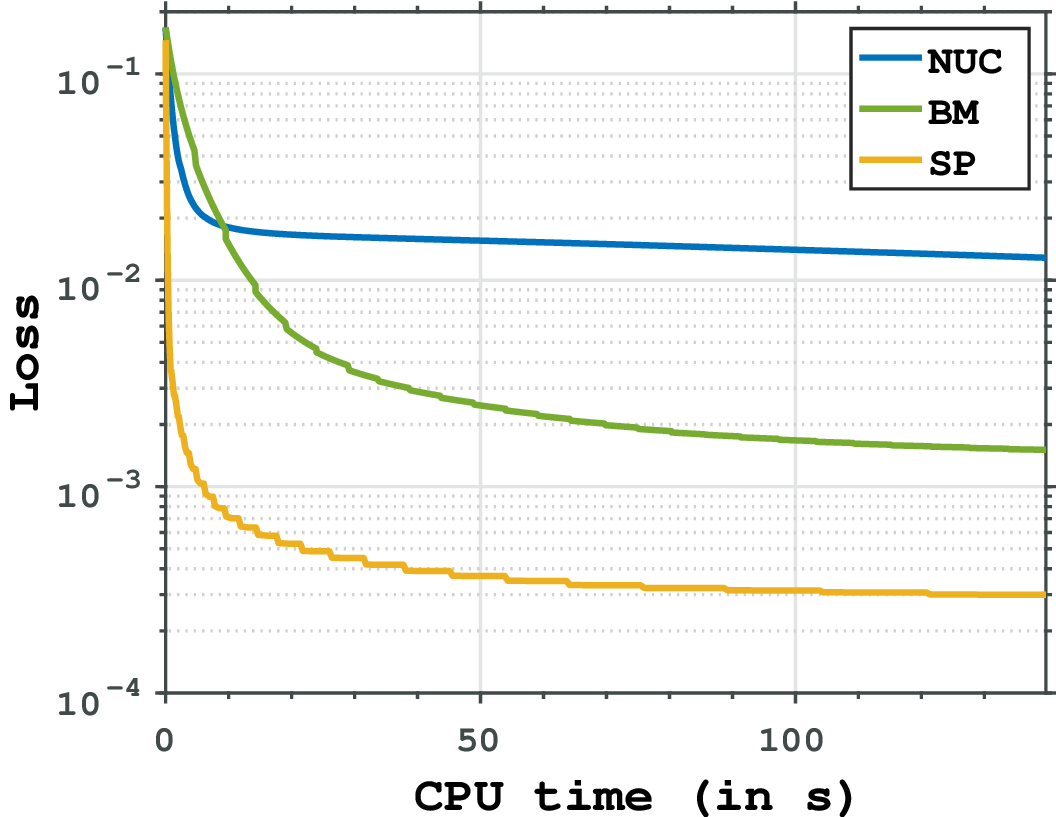}
        \label{fig:nfig2}
    }
    \subfigure[Hankel Singular Values vs Time]{
        \includegraphics[width=0.48\textwidth]{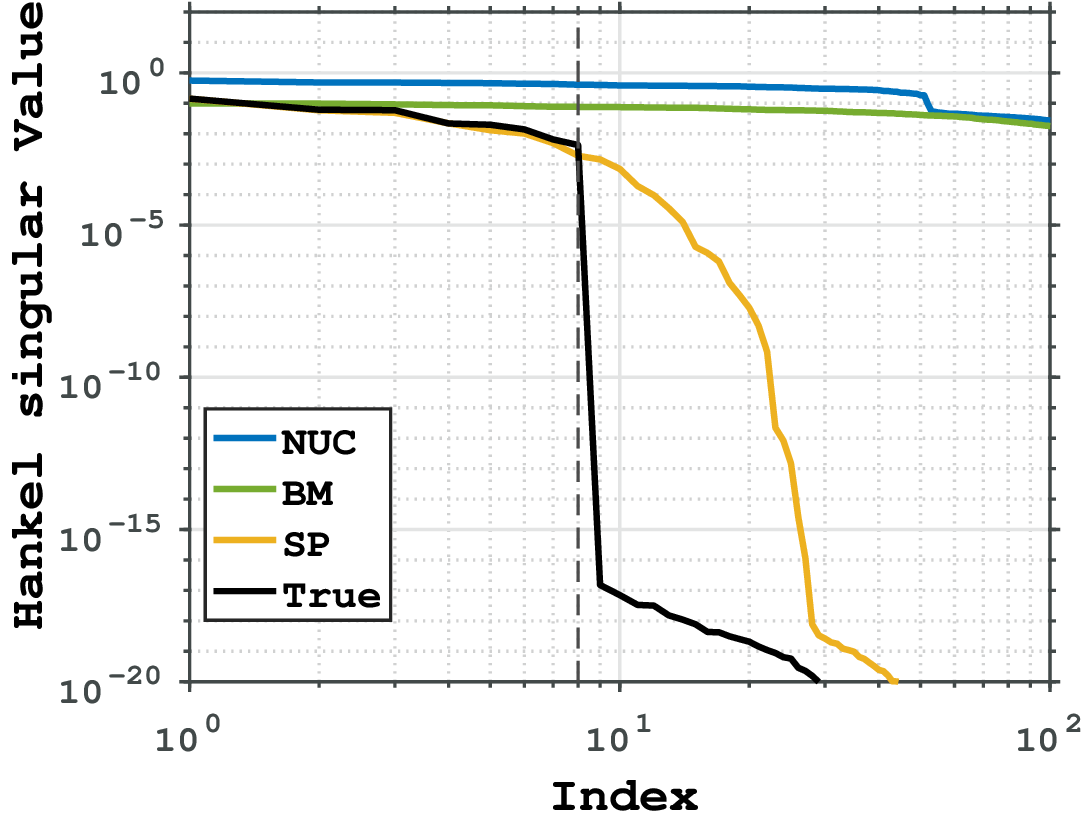}
        \label{fig:nfig3}
    }
    \subfigure[Polar vs Time]{
        \includegraphics[width=0.48\textwidth]{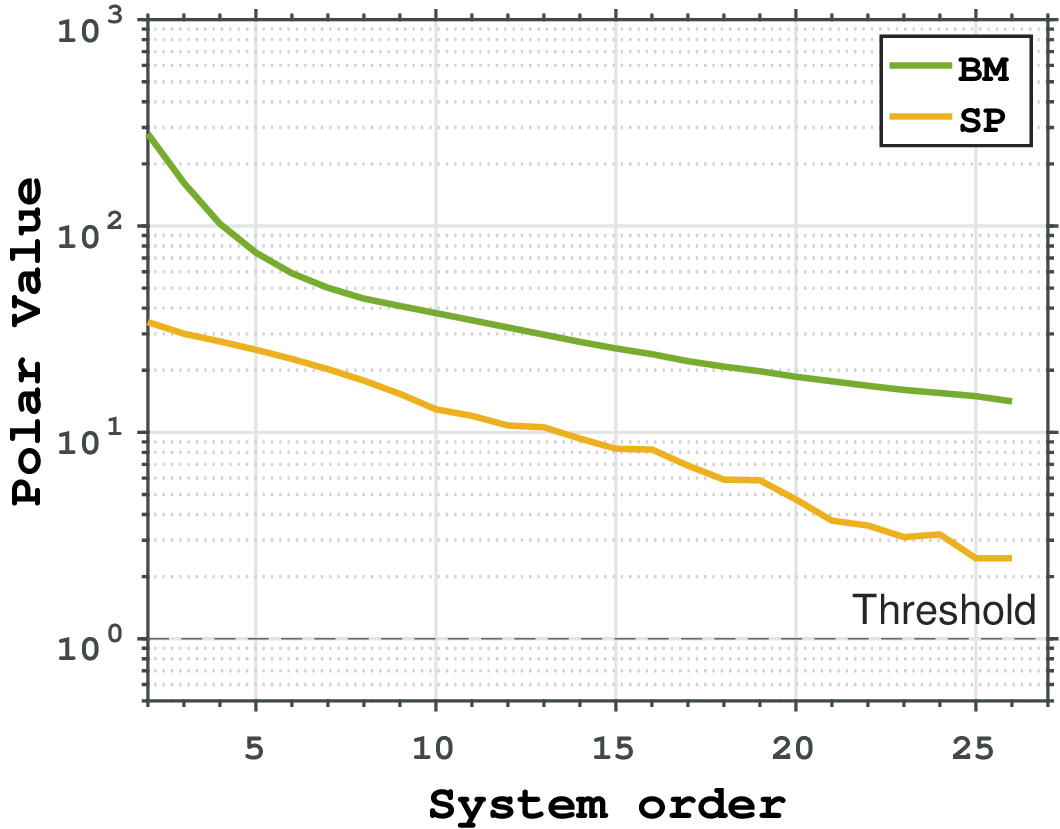}
        \label{fig:nfig4}
    }

    \caption{Overview of various metrics for a non-diagnosable systems with sample size, $N=500$, and trajectory length $T=102$.\textsf{NUC} refers to \eqref{eq:P1_1},
    \textsf{BM} refers to \eqref{eq:P2}, and \textsf{SP} refers to \eqref{eq:P1}.}
    \label{fig:non_diag_metric}
\end{figure}

\end{document}

%% file: notation.tex
\def\a{{\mathbf a}}
\def\b{{\mathbf b}}
\def\c{{\mathbf c}}

\def\i{{\mathbf i}}

\def\u{{\mathbf u}}
\def\v{{\mathbf v}}
\def\w{{\mathbf w}}
\def\x{{\mathbf x}}
\def\y{{\mathbf y}}
\def\z{{\mathbf z}}

\def\C{{\cal C}}

\def\F{{\cal F}}

\def\H{{\cal H}}

\def\N{{\cal N}}
\def\P{{\cal P}}

\def\R{{\mathbb R}}

\def\l{\lambda}

\def\nm{\Vert}

\renewcommand{\and}{\mbox{$\wedge$}}

\newcommand{\bc}{\begin{center}}
\newcommand{\ec}{\end{center}}
\newcommand{\be}{\begin{equation}}
\newcommand{\ee}{\end{equation}}
\newcommand{\bd}{\begin{displaymath}}
\newcommand{\ed}{\end{displaymath}}
\newcommand{\ba}{\begin{array}}
\newcommand{\ea}{\end{array}}
\newcommand{\ben}{\begin{enumerate}}
\newcommand{\een}{\end{enumerate}}
\newcommand{\bit}{\begin{itemize}}
\newcommand{\eit}{\end{itemize}}
\newcommand{\beq}{\begin{eqnarray}}
\newcommand{\eeq}{\end{eqnarray}}
\newcommand{\btab}{\begin{tabular}}
\newcommand{\etab}{\end{tabular}}
\newcommand{\mathbfig}{\begin{figure}}
\newcommand{\efig}{\end{figure}}
\newcommand{\btp}{\begin{tikzpicture}}
\newcommand{\etp}{\end{tikzpicture}}

\newcommand{\argmin}{\operatornamewithlimits{arg~min}}

\newcommand{\nmm}[1]{ \nm #1 \nm }

\newcommand{\nmF}[1]{ \nm #1 \nm_F }

\newcommand{\IP}[2]{ \langle #1 , #2 \rangle }

\def\nmsl1{\nm_{{\rm SL1}}}

\def\vecc{{\mathsf{vec}}}

\def\diag{{\mathsf{diag}}}
\def\exp{{\mathsf{exp}}}


%% file: algorithms_details.tex
\begin{algorithm}
\caption{Accelerated proximal gradient descent for Hankel Nuclear Norm Minimization \eqref{eq:P1_1}}
\label{alg:P1_1}
\SetKwInOut{Input}{Input}
\SetKwInOut{Output}{Output}

\Input{
    $Y_i \in \mathbb{R}^{n_y \times (2L+2) \times N}$: Output data tensor,
    $U_i \in \mathbb{R}^{n_u \times (2L+2) \times N}$: Input data tensor,\\
    $\lambda$: Regularization parameter,
    $\text{lr}$: Learning rate,
    $\mu$: Momentum rate,\\
    $G_{\text{init}} \in \mathbb{R}^{n_y \times (n_u \cdot (2L+1))}$: Initialization,\\
    $\text{max\_iter}$: Maximum iterations.
}
\Output{
    $\hat{n}_x, (\hat{A}, \hat{B}, \hat{C})$: Estimated system parameters.
}

\BlankLine
\textbf{Initialization:}\\
$G_0 \gets G_{\text{init}}$\\
$L(G) := \frac{1}{2N}\sum_{i=1}^{N}\nmm{\y_{2L+2}^{(i)} - G\vecc(U_{1:2L+1}^{(i)})}_2^2$.
\BlankLine
\For{$k \gets 1$ \KwTo $\text{max\_iter}$}{
$G_k \gets G_{k-1} - \text{lr} \cdot \nabla_{G}L(G_{k-1})$\\
$L, \Sigma, R \gets \text{SVD}(\mathcal{H}(G_{k}))$\\
$G_k \gets \mathcal{H}^{\dagger}\left(L(\Sigma - \text{lr} \cdot \l \cdot I)R^T\right) + \mu \cdot (G_k - G_{k-1})$
}
$\hat{n}_x, (\hat{A}, \hat{B}, \hat{C}) \gets \textbf{Ho-Kalman}(G_k)$,\\
\textbf{Return:} $\hat{n}_x, (\hat{A}, \hat{B}, \hat{C})$.
\end{algorithm}

\begin{algorithm}
\caption{Polyak's gradient descent for program \eqref{eq:P2}}\label{alg:P2}
\SetKwInOut{Input}{Input}
\SetKwInOut{Output}{Output}

\Input{
    $Y_i \in \mathbb{R}^{n_y \times (2L+2) \times N}$: Output data tensor,
    $U_i \in \mathbb{R}^{n_u \times (2L+2) \times N}$: Input data tensor,\\
    $\lambda$: Regularization parameter,
    $\text{lr}$: Learning rate,
    $\mu$: Momentum rate,\\
    $R_{\text{init}}$: Initial system order, $R_{\text{max}}$: Maximum system order\\
    $V_{\text{init}} \in \R^{(L+1)n_y \times R_{\text{init}}}$, $Z_{\text{init}}
    \in \R^{(L+1)n_u \times R_{\text{init}}}$: Initialization,\\
    $\text{max\_iter}$: Maximum iterations.
}
\Output{
    $\hat{n}_x, (\hat{A}, \hat{B}, \hat{C})$: Estimated system parameters.
}

\BlankLine
\textbf{Initialization:}\\
$V_0, V_{-1} \gets V_{\text{init}}$\\
$Z_0, Z_{-1} \gets Z_{\text{init}}$\\
$n_x \gets R_{\text{init}}$\\
$L(V, Z) := \frac{1}{2N}\sum_{i=1}^{N}\nmm{\y_{2L+2}^{(i)} - \mathcal{H}^{\dagger}(VZ^T)\vecc(U_{1:2L+1}^{(i)})}_2^2 + \frac{\l}{2}\left[\nmF{V}^2 + \nmF{Z}^2\right]$.\\
$M(V, Z) := \frac{1}{N\l}\sum_{i=1}^{N}\left(Y_i - \mathcal{H}^{\dagger}(VZ^T)\vecc(U_{1:2L+1}^{(i)}\right)
^T\vecc(U_{1:2L+1}^{(i)})$
\BlankLine
\While{$\nmm{\mathcal{H}^{\dagger}(V_0, Z_0)}_2 > 1$ and $\hat{n}_x \leq R_{\text{max}}$}{
\For{$k \gets 1$ \KwTo $\text{max\_iter}$}{
$V_k \gets V_{k-1} - \text{lr} \cdot \nabla_{V}L(V_{k-1}, Z_{k-1}) + \mu \cdot (V_{k-1} - V_{k-2})$\\
$Z_k \gets Z_{k-1} - \text{lr} \cdot \nabla_{Z}L(V_{k-1}, Z_{k-1}) + \mu \cdot (Z_{k-1} - Z_{k-2})$
}
$[L, \Sigma, R] \gets \text{SVD}(\mathcal{H}(M(V_{\text{max\_iters}}, Z_{\text{max\_iters}})))$\\
$\v^* \gets L(:, 1)$\\
$\z^* \gets R(:, 1)$\\
$V_0, V_{-1} \gets \begin{bmatrix} V_{\text{max\_iters}} & \v^* \end{bmatrix}$\\
$Z_0, Z_{-1} \gets \begin{bmatrix} Z_{\text{max\_iters}} & \z^* \end{bmatrix}$\\
$\hat{n}_x \gets \hat{n}_x + 1$
}
$\hat{n}_x \gets n_x$,\\
$\hat{C} \gets \left[V_{\text{max\_iters}}\right]_{1:n_y, 1:\hat{n}_x}$\\
$\hat{B} \gets \left[Z_{\text{max\_iters}}^T\right]_{1:n_u, 1:\hat{n}_x}$\\
$\hat{A} \gets \left(\left[V_{\text{max\_iters}}^T\right]_{(n_u+1):2\cdot n_u, 1:\hat{n}_x}\right)^{\dagger}\left[V_{\text{max\_iters}}\right]_{1:n_y, 1:\hat{n}_x}$\\
\textbf{Return:} $\hat{n}_x, (\hat{A}, \hat{B}, \hat{C})$.
\end{algorithm}

\begin{algorithm}
\caption{Polyak's gradient descent for program \eqref{eq:P1}}\label{alg:P1}
\SetKwInOut{Input}{Input}
\SetKwInOut{Output}{Output}

\Input{
    $Y_i \in \mathbb{R}^{n_y \times (2L+2) \times N}$: Output data tensor,
    $U_i \in \mathbb{R}^{n_u \times (2L+2) \times N}$: Input data tensor,\\
    $\lambda$: Regularization parameter,
    $\text{lr}$: Learning rate,
    $\mu$: Momentum rate,\\
    $R_{\text{init}}$: Initial system order, $R_{\text{max}}$: Maximum system order\\
    $\a_{\text{init}} \in \R^{R_{\text{init}}}$,
    $B_{\text{init}} \in \R^{R_{\text{init}} \times n_u}$,
    $C_{\text{init}} \in \R^{n_y \times R_{\text{init}}}$: Initialization,\\
    $\text{max\_iter}$: Maximum iterations.
}
\Output{
    $\hat{n}_x, (\hat{A}, \hat{B}, \hat{C})$: Estimated system parameters.
}

\BlankLine
\textbf{Initialization:}\\
$\a_0, \a_{-1} \gets \a_{\text{init}}$\\
$B_0, B_{-1} \gets B_{\text{init}}$\\
$C_0, C_{-1} \gets C_{\text{init}}$\\
$n_x \gets R_{\text{init}}$\\
$L(\a, B, C)$ := $\frac{1}{4N(L+1)}\sum_{i=1}^{N}\nmm{\vecc(Y^{(i)}) - G(\diag(\a), B, C, 0)\vecc(U_{1:2L+1}^{(i)})}_2^2$ + ${\l}\sum_{j=1}^{n_x}\gamma(a_j)\nmm{B_{j, :}}_2\nmm{C_{:, j}}_2$.\\
$M(a', \a, B, C) := \frac{1}{2N(L+1)\l\gamma(a')}\sum_{i=1}^{N}\left[Y_i - \sum_{j=1}^{n_x}\c_j\b_j^TU_i'P^T(a_j)\right]P(a'){U_i'}^{T}$
\BlankLine
\While{$\nmm{\mathcal{H}^{\dagger}(V_0, Z_0)}_2 > 1$ and $\hat{n}_x \leq R_{\text{max}}$}{
\For{$k \gets 1$ \KwTo $\text{max\_iter}$}{
$\a_k \gets \a_{k-1} - \text{lr} \cdot \nabla_{\a}L(\a_{k-1}, B_{k-1}, C_{k-1}) + \mu \cdot (\a_{k-1} - \a_{k-2})$\\
$B_k \gets B_{k-1} - \text{lr} \cdot \nabla_{B}L(\a_{k-1}, B_{k-1}, C_{k-1}) + \mu \cdot (B_{k-1} - B_{k-2})$\\
$C_k \gets C_{k-1} - \text{lr} \cdot \nabla_{C}L(\a_{k-1}, B_{k-1}, C_{k-1}) + \mu \cdot (C_{k-1} - C_{k-2})$
}
$a^* \gets \textbf{Golden-section Search}(M(\cdot, \a_{\text{max\_iters}}, B_{\text{max\_iters}}, C_{\text{max\_iters}}))$\\
$[L, \Sigma, R] \gets \text{SVD}(M(\cdot, \a_{\text{max\_iters}}, B_{\text{max\_iters}}, C_{\text{max\_iters}}))$\\
$\c^* \gets L(:, 1)$\\
$\b^* \gets R(:, 1)$\\
$\a_0, \a_{-1} \gets \begin{bmatrix} \a_{\text{max\_iters}}^T & a^* \end{bmatrix}$\\
$B_0, B_{-1} \gets \begin{bmatrix} B_{\text{max\_iters}}^T & \b^* \end{bmatrix}^T$\\
$C_0, C_{-1} \gets \begin{bmatrix} C_{\text{max\_iters}} & \c^* \end{bmatrix}$\\
$\hat{n}_x \gets \hat{n}_x + 1$
}
$\hat{n}_x \gets n_x$, $\hat{A} \gets \diag(\a_{\text{max\_iters}})$ ,$\hat{B} \gets B_{\text{max\_iters}}$, $\hat{C} \gets C_{\text{max\_iters}}$.\\
\textbf{Return:} $\hat{n}_x, (\hat{A}, \hat{B}, \hat{C})$.
\end{algorithm}